\documentclass[journal]{IEEEtran}

\usepackage{amsmath,amssymb,amsfonts}

\usepackage[english]{babel}
\usepackage{algorithm,algorithmic}
\usepackage{graphicx}
\usepackage{textcomp}
\usepackage{color}
\usepackage{lipsum}
\usepackage{epstopdf}
\usepackage{amsopn}
\usepackage{amsthm}
\usepackage{enumerate}
\usepackage{cancel}
\usepackage{mathrsfs}
\usepackage{mathdots}
\usepackage{euscript}
\usepackage{amscd}
\usepackage{cite}
\usepackage{placeins}
\usepackage[latin1]{inputenc}
\usepackage{tikz}
\usetikzlibrary{snakes,arrows,shapes}

\usepackage{subcaption}

\newtheorem{theorem}{Theorem}

\newtheorem{proposition}{Proposition}

\theoremstyle{definition}
\newtheorem{definition}{Definition}
\newtheorem{example}{Example}
\newtheorem{problem}{Problem}
\newtheorem{assumption}{Assumption}
\theoremstyle{remark}\newtheorem{remark}{Remark}

\newcommand{\bmat}{\begin{bmatrix}}
\newcommand{\emat}{\end{bmatrix}}

\newcommand{\innerprod}[2]{\langle{#1},\,{#2}\rangle}

\newcommand{\var}{\mathop{\rm var}}
\newcommand{\cov}{\mathop{\rm cov}}

\DeclareMathOperator{\diag}{diag}

\newcommand{\E}{{\mathbb E}}
\newcommand{\Rbb}{\mathbb R}
\newcommand{\Hbb}{\mathbb H}

\newcommand{\Zbb}{\mathbb Z}

\newcommand{\Nbb}{\mathbb N}
\newcommand{\Tbb}{\mathbb T}
\newcommand{\yb}{\mathbf  y}
\newcommand{\zb}{\mathbf  z}

\newcommand{\ab}{\mathbf a}
\newcommand{\kb}{\mathbf k}
\newcommand{\bb}{\mathbf  b}
\newcommand{\cb}{\mathbf  c}

\newcommand{\pb}{\mathbf  p}

\newcommand{\tb}{\mathbf t} 
\newcommand{\qb}{\mathbf q}  

\newcommand{\mb}{\mathbf  m}
\newcommand{\lb}{\boldsymbol{\ell}}

\newcommand{\Cb}{\mathbf C}

\newcommand{\Fb}{\mathbf F}

\newcommand{\Ib}{\mathbf I}

\newcommand{\Kb}{\mathbf K}

\newcommand{\Nb}{\mathbf N}

\newcommand{\Ub}{\mathbf U}

\newcommand{\Yb}{\mathbf Y}

\newcommand{\thetab}{\boldsymbol{\theta}}

\newcommand{\zerob}{\boldsymbol{0}}

\renewcommand{\d}{\mathrm{d}}

\renewcommand{\Re}{\mathrm{Re}}
\renewcommand{\Im}{\mathrm{Im}}

\newcommand{\F}{\mathrm{F}}

\newcommand{\nn}{\nonumber}

\begin{document}

\title{On the Statistical Consistency of a Generalized Cepstral Estimator}

%\title{On the Generalized Cepstral Estimation with Application to \bin{Cascade} System Identification\\ {\mz On the Consistency of Generalized Cepstral Estimators}\\ 
%\footnotesize{\mz [MATTIA: I think it is dangerous to use both the words application and cascade. Regarding the former: many people is sensitive to this keyword and they understand it as ``we have applied something in a real experiment (= real data) or at least we show that something can be useful in a precise case study (= identification of a robotic system in a particular configuration, etc.)'' The risk is that the reviewers say that we try to sell what we don't really do. Regarding the latter word: the title seems to suggest that our main aim is to provide a theory that gives some improvement in the field of cascade systems. So, I would expect that the associate editor invite at least one reviewer expert on that field. The risk is that they ask us to make comparison with existing methods or say that our estimator does not provide any novelty because the problem has been already solved (and the multidimensional is not interesting in control)] }}

\author{Bin Zhu and Mattia Zorzi% <-this % stops a space
\thanks{B.~Zhu is with the School of Intelligent Systems Engineering, Sun Yat-sen University, Waihuan East Road 132, 510006 Guangzhou, China (email: \texttt{zhub26@mail.sysu.edu.cn}).}% <-this % stops a space
\thanks{M.~Zorzi is with the Department of Information Engineering, University of Padova, Via Giovanni Gradenigo, 6b, 35131 Padova, Italy (email: \texttt{zorzimat@dei.unipd.it}).}
\thanks{This work was supported in part by the National Natural Science Foundation of China under the grant number 62103453 and the ``Hundred-Talent Program'' of Sun Yat-sen University. Corresponding author B.~Zhu. Tel. +86 14748797525. Fax +86(20) 39336557.}%
}

%\markboth{IEEE Transactions on }%
%{}

\maketitle

%{\mz
\begin{abstract}
We consider the problem to estimate the generalized cepstral coefficients of a stationary stochastic process or stationary multidimensional random field. It turns out that a naive version of the periodogram-based estimator for the generalized cepstral coefficients is not consistent. We propose a consistent estimator for those coefficients. Moreover, we show that the latter can be used in order to build a consistent estimator for a particular class of cascade linear stochastic systems. 
\end{abstract}
%}


\section{Introduction}\label{sec:intro}

%general introduction of the problem area

Given a stationary time series $y_0, y_1, \dots, y_{N-1}$, the estimation of some second-order statistics of the series, e.g., the covariances $c_k := \E\, [y_{t+k} \,y_t^*]$ where $\E$ denotes mathematical expectation, has been a basic problem in the fields of signal processing, identification, and systems theory.
In particular, the covariance estimates are commonly used in practical applications where {high resolution spectral estimators are needed}. More specifically, we would like to mention the line of research on \emph{rational covariance extension}, see for instance \cite{BGL-98,SIGEST-01,Georgiou-01,byrnes2001cepstral,FPR-08,FRT-11,RFP-09,ringh2018multidimensional,zhu2018wellposed,FMP-12,Z-14,Zhu-Baggio-19}. 
%\bin{[REVIEW the line of research on rational covariance extension...MAYBE BETTER to put in the Introduction.]}

Let $\Tbb:=[0,2\pi)$ be the domain for the angular frequency. It follows from {the} classical theory of stationary processes that the covariances of $y_t$ are the Fourier coefficients of the \emph{power spectral density} $\Phi(\theta)$ which is a \emph{nonnegative} function on $\Tbb$, that is
\begin{equation}\label{cov_def}
c_k = \int_{\Tbb} e^{ik\theta} \Phi(\theta) \frac{\d\theta}{2\pi}, \quad k\in \Zbb.
\end{equation}
In view of the above relation, the covariances $c_k$ can be viewed as \emph{moments} (a term from Mechanics) of the power spectrum $\Phi$ where the complex exponentials $\{e^{ik\theta}\}_{k\in\Zbb}$ play the role of \emph{basis functions}. Another important instance of moments is the cepstral coefficients
\begin{equation}\label{cep_def}
m_k := \int_{\Tbb} e^{ik\theta} \log\Phi(\theta) \frac{\d\theta}{2\pi}, \quad k\in \Zbb,
\end{equation}
assuming that the integration gives a finite number, e.g., when $\log\Phi\in L^1(\Tbb)$. These $m_k$'s find applications notably in speech processing, see e.g., \cite{ephraim1999second,ephraim2005second}.
In addition, cepstral coefficients represent a fundamental tool in covariance extension problems in order to obtain rational spectral densities with zeros, see \cite{enqvist2004aconvex,byrnes2001cepstral}.

Inspired by the integral representations \eqref{cov_def} and \eqref{cep_def}, a natural choice of the estimators $\hat{c}_k$ and $\hat{m}_k$ is to replace the true spectrum $\Phi$ with the \emph{periodogram} $\hat{\Phi}$ (a simple spectral estimator) and to discretize the integral into a finite summation.
An important question is whether these estimators are statistically consistent, that is: does $\hat{c}_k$ (and $\hat{m}_k$) converge in a certain stochastic sense to the true value $c_k$ (and $m_k$, respectively) as the length $N$ of the available time series goes to infinity?
Such a consistency property is useful in e.g., evaluating the consistency of those moment-based spectral estimators when the model class is properly chosen, see e.g., {\cite{Zhu-Zorzi-2021-cepstral,RFP-10-wellposedness}}. 

%what has been done in the literature
For the covariance estimator $\hat{c}_k$, an affirmative answer to the consistency question can be found in \cite[Subsec.~5.3.3]{priestley1981spectral} under some Gaussian assumptions, which is somewhat surprising because the periodogram $\hat{\Phi}$ alone is \emph{not} a consistent estimator of the true spectrum $\Phi$ since its variance (at each frequency $\theta$) does not converge to zero as $N\to \infty$, see e.g., \cite[p.~425]{priestley1981spectral}. Consistency of $\hat{c}_k$ comes out as a result of the averaging operation in the discrete version of \eqref{cov_def} involving $\hat{\Phi}$.
The consistency question for the cepstral estimator $\hat{m}_k$ is significantly harder than that for $\hat{c}_k$, because the latter estimator is linear in $\hat{\Phi}$ while the former one involves the logarithm (which is obviously nonlinear). The answer is still affirmative and can be found in \cite{ephraim1999second} under the assumption that the spectral components $Y_\ell$ (to be defined later in Section~\ref{sec:prob}) are independent Gaussian random variables.

%what is missing, contribution of the current paper
In the case of multidimensional random fields, however, the {classical} cepstral coefficients \eqref{cep_def} do not always guarantee that the covariance extension problem admits {a} solution, see \cite{KLR-16multidimensional,ringh2018multidimensional}. 
In order to overcome this difficulty, in {the} recent papers \cite{Zhu-Zorzi-21,Zhu-Zorzi-2021-cepstral,ringh2015multidimensional} of the authors, the \emph{generalized cepstral coefficients} \cite{tokuda1990generalized}, which are the moments of  a power function of the spectrum, have been used to set up a rational covariance extension problem which admits {as the solution a} rational multidimensional spectral density with zeros.
% solution and such a solution is rational multidimensional spectral density with zeros.
To our surprise, the periodogram-based estimator for the generalized cepstral coefficients appeared to be inconsistent when we were carrying out simulation studies. However, it turns out in the later development of the current paper that the generalized cepstral estimator using $\hat{\Phi}$ is indeed \emph{consistent up to a constant multiplicative factor} under the same Gaussian assumption for the usual cepstral estimator in \cite{ephraim1999second}. In other words, the unwindowed periodogram suffices to produce a consistent estimator of the generalized cepstral coefficients and all one needs to do is to suitably rescale the computed quantities. To the best of our knowledge, such a consistency question has never been addressed before in the literature. Finally, we show that our result can be used to build a consistent estimator for a class of \emph{cascade} linear stochastic systems.
%INCLUDE the motivation for cascade system identification \cite{wahlberg2009variance,sandberg2014approximative}...
As explained in e.g., \cite{wahlberg2009variance,sandberg2014approximative}, cascade systems are very common in process industry, signal processing, and other engineering applications.

This paper is organized as follows. Section~\ref{sec:review_DFT} reviews the DFT and sets up the notation.
Section~\ref{sec:prob} formulates the problem of generalized cepstral estimation and its statistical consistency.
Section~\ref{sec:independ} shows the consistency of the generalized cepstral estimator under the assumption of independent Gaussian spectral components.
Section~\ref{sec:depend} deals with the more general case of correlated spectral components.
Section~\ref{sec:multidim} extends the consistency result of the generalized cepstral estimation to multidimensional random fields.
Section~\ref{sec:sysid} describes an application of the generalized cepstral estimator to cascade system identification, for which some simulation results are presented in Section~\ref{sec:sims}.
Finally, Section~\ref{sec:conclusions} draws the conclusions.

%\subsection*{Notation and convention}
%
%\bin{See if necessary to specify the DFT convention HERE...}
%Boldface symbols denote random quantities.
%For a  square summable sequence $y$ of complex numbers, we take the definition of the discrete-time Fourier transform (DTFT) to be the following
%\begin{equation*}
%\begin{split}
%\Fcal:\, \ell^2 & \to L^2[-\pi,\pi]\\
%y & \mapsto \hat{y}(\omega):= \sum_{t\in\Zbb} y(t) e^{-i t \omega},
%\end{split}
%\end{equation*}
%where the convergence of the Fourier series is understood in $L^2$ norm. The inverse transform is given by
%\begin{equation*}
%%\begin{split}
%\Fcal^{-1}:\, \hat{y} \mapsto y(t) := \frac{1}{2\pi} \int_{-\pi}^{\pi} e^{i t \omega} \hat{y}(\omega) \d\omega.
%%\end{split}
%\end{equation*}
%The $\ell^2$ norm of $y$ is known as the energy of the signal.

\section{Review of the Discrete Fourier Transform}\label{sec:review_DFT}

The aim of this short section is to clarify the notation of the DFT that is used throughout the paper. 

The standard DFT (in the signal processing literature) of a finite complex-valued sequence $y_0,\dots,y_{N-1}$ of length $N$ is defined as
\begin{equation}\label{def_DFT}
Y_\ell = \sum_{t=0}^{N-1} y_t\, e^{-it\frac{2\pi}{N}\ell}, \quad \ell=0,\dots,N-1.
\end{equation}
The inverse DFT is known as
\begin{equation}\label{def_IDFT}
y_t = \frac{1}{N} \sum_{\ell=0}^{N-1} Y_\ell \,e^{it\frac{2\pi}{N}\ell},\quad t=0,\dots,N-1.
\end{equation}
The above two expression can be rewritten in matrix forms as
\begin{equation}\label{vec_notation}
\Yb = \Fb \yb \ \text{ and } \ \yb = \frac{1}{N} \Fb^* \Yb,
\end{equation}
in which $\Yb=[Y_0, \dots, Y_{N-1}]^\top$, $\yb=[y_0, \dots, y_{N-1}]^\top$, and
\begin{equation}\label{DFT_mat}
\Fb = \bmat \omega_N^{0\cdot 0} & \omega_N^{0\cdot 1} & \cdots & \omega_N^{0\cdot (N-1)}\\ 
 \omega_N^{1\cdot 0} & \omega_N^{1\cdot 1} & \cdots & \omega_N^{1\cdot (N-1)}\\ 
 \vdots & \vdots & \ddots & \vdots \\
 \omega_N^{(N-1)\cdot 0} & \omega_N^{(N-1)\cdot 1} & \cdots & \omega_N^{(N-1)\cdot (N-1)}\\ 
\emat
\end{equation}
is the DFT matrix where $\omega_N = e^{-i\frac{2\pi}{N}}$ is a primitive $N$-th root of unity. It follows immediately that $\Fb^{-1} = \frac{1}{N}\Fb^*$. Notice also that $\Fb$ is a complex symmetric matrix so that $\Fb^*=\overline{\Fb}$ (elementwise complex conjugate).

%Given a sequence of length $N$, define its DFT...

%self-adjoint/unitary operator?

%[CHECK $K$-point DFT and invert only the first $N$ points. Probably this is different from the $N$-point DFT. DECIDE whether we need to choose $K=2N-1$ or $K=N$. \bin{IT SEEMS that if we want the circulant covariance matrix interpretation, we will need $K=N$. In that case, we should start with a complex process...}]

%matrix notation for the DFT and its inverse, DFT matrix which is unitary up to a constant $K$.
%\begin{equation}
%\bmat \bar{Y}_0 \\ \vdots \\ \bar{Y}_{N-1} \emat = \frac{1}{\sqrt{N}} \verb|fft| ( \bmat y(0) \\ \vdots \\ y(N-1) \emat )
%\end{equation}

\section{Problem Statement}\label{sec:prob}

%\bin{[LEAVE the periodogram here and MOVE the definitions of moments to the Introduction...]}

In order to introduce the generalized cepstral estimation problem, we need to first review some basic notions of the periodogram and covariance estimation.

Consider a zero-mean second-order stationary complex-valued discrete-time random signal $y_t$ (so that $t\in\Zbb$). Suppose that we are given a finite number of random samples $y_0, y_1, \dots, y_{N-1}$. 
Taking the (unilateral) finite Fourier transform, we obtain
\begin{equation}\label{spec_comp_theta}
Y(\theta) = \sum_{t=0}^{N-1} y_t \, e^{-it\theta},
\end{equation}
a random variable that depends on the angular frequency $\theta\in\Tbb=[0,2\pi)$. Notice that the frequency interval $[0,2\pi)$ is understood as the quotient group $\Rbb / 2\pi\Zbb$, so that the function $Y(\theta)$ is $2\pi$-periodic on $\Rbb$. In the spectral theory of stationary processes, $\Tbb$ is often identified as the unit circle on the complex plane, and one writes $Y(e^{i\theta})$ instead \cite{LP15}. Here we keep the simpler notation $Y(\theta)$. The (unwindowed) periodogram, which is a first estimate of the power spectrum $\Phi(\theta)$ of the random process $y_t$, is defined as
\begin{equation}
\hat{\Phi}(\theta) := \frac{1}{N} |Y(\theta)|^2.
\end{equation}
%Due to assumption that the process $y(t)$ is real, its power spectrum is an even function, meaning $\Phi(-\theta) = \Phi(\theta)$, or equivalently, symmetric with respect to the line $\theta=\pi$. In this sense, one often regard $\Phi$ or its estimate $\hat{\Phi}$ as a function on the half interval $[0,\pi)$ where the spectral content of the signal $y$ concentrates.
%[THIS interpretation has nothing to do with $K$!] 
In addition, the periodogram $\hat{\Phi}$ admits a \emph{correlogram} interpretation as the \emph{bilateral} finite Fourier transform of the standard biased covariance estimates \cite{stoica2005spectral}, that is,
\begin{equation}\label{Phi_hat_correlo}
\hat{\Phi}(\theta) = \sum_{k=-N+1}^{N-1} \hat{c}_k e^{-ik\theta},
\end{equation}
where
\begin{equation}\label{cov_est_average}
\hat{c}_k = \frac{1}{N} \sum_{t=0}^{N-1-k} y_{t+k} \, y_t^*, \quad k=0, 1, \dots, N-1
\end{equation}
and $\hat{c}_{-k} = \hat{c}_k^*$. Here $(\cdot)^*$ means taking complex conjugate when applied to a complex number. 

In practice, the periodogram $\hat{\Phi}$ is evaluated on a regular grid of $\Tbb$ of size $N$ equal to the length of the time series $y_t$. In other words, the frequency interval $[0,2\pi)$ is partitioned into $N$ subintervals of equal length $2\pi/N$, and the grid points are collected in the set
\begin{equation}\label{T_grid}
\Tbb_{N} := \left\{ \frac{2\pi}{N}\ell \,:\, \ell = 0,1,\dots,N-1 \right\}.
\end{equation}
Let us write
%\begin{equation}
$Y_\ell = Y({2\pi\ell}/{N})$, and similarly $\hat\Phi_\ell = \hat\Phi({2\pi\ell}/{N})$
%\end{equation}
for simplicity.
The random variables $\{Y_\ell\}_{\ell=0}^{N-1}$ are termed \emph{spectral components} of the process $y_t$ \cite{ephraim1999second}.
%Moreover, we make the convenient choice $K\geq 2N-1$ which is sufficient to guarantee that
%[\bin{PAY more attention to this condition!} CHECK \cite{ephraim1999second} for the reason for this choice. It appears that $K=N$ is fine for the proof and also important for the periodic process/circulant covariance matrix interpretation...] 
%It is well known  that under the assumption $K\geq 2N-1$, 
The covariance estimates can then be computed from the discrete spectrum $\{\hat{\Phi}_\ell\}_{\ell=0}^{N-1}$ via the inverse DFT (see also Sec.~\ref{sec:review_DFT})
\begin{equation}\label{cov_estimates}
\hat{c}_k \approx \frac{1}{N} \sum_{\ell=0}^{N-1} e^{ik\frac{2\pi}{N}\ell} \hat{\Phi}_\ell,
\end{equation}
%[\bin{THIS is more of a practical convenient method as in principle, fewer than $K$ points in the discrete spectrum is needed due to the symmetry $\hat{c}_{-k} = \hat{c}_k$.}]
which is seen as an approximation to the Fourier integral
\begin{equation}\label{cov_integral}
\hat{c}_k = \int_{\Tbb} e^{ik\theta} \hat{\Phi}(\theta) \frac{\d\theta}{2\pi}.
\end{equation}

\begin{remark}
	%[\bin{COMMENT on the $K$-point FFT.}]
	Alternatively, we can evaluate the periodogram on a denser grid of size $K\geq 2N-1$. The motivation for doing so is that such a choice of the grid size  	
	%make the convenient choice $K\geq 2N-1$ which 
	is sufficient to guarantee that
	%[\bin{PAY more attention to this condition!} CHECK \cite{ephraim1999second} for the reason for this choice. It appears that $K=N$ is fine for the proof and also important for the periodic process/circulant covariance matrix interpretation...] 
	%It is well known  that under the assumption $K\geq 2N-1$, 
	the covariance estimates can be \emph{exactly} recovered from the discrete spectrum $\{\hat{\Phi}_\ell\}_{\ell=0}^{K-1}$ via the formula %inverse DFT (see also Sec.~\ref{sec:review_DFT})
	\begin{equation}\label{cov_estimates_K}
	\hat{c}_k = \frac{1}{K} \sum_{\ell=0}^{K-1} e^{ik\frac{2\pi}{K}\ell} \hat{\Phi}(2\pi\ell/K),
	\end{equation}
	%{\mz [MATTIA: I think there is a glitch between the notations in (11) (12), i.e. $\hat c_k$ is used to denote two different things. \bin{Bin: No, \eqref{cov_integral} and \eqref{cov_estimates_K} represent exactly the same quantity, because the Fourier coefficients of $\hat{\Phi}$ can be recovered via the DFT with a sufficient number of equidistant samples.}]}
	%[\bin{CHANGE $K$ to $N$ in the above equation. COMMENT on the $K$-point FFT.}]
	%[\bin{THIS is more of a practical convenient method as in principle, fewer than $K$ points in the discrete spectrum is needed due to the symmetry $\hat{c}_{-k} = \hat{c}_k$.}]
	in contrast to the usual Fourier integral \eqref{cov_integral}.
	Computation via \eqref{cov_estimates_K} is more efficient than the direct time average \eqref{cov_est_average} thanks to the FFT routines. However, we also want to point out that as $N$ (hence also $K$) tends to infinity, the difference between \eqref{cov_estimates} and \eqref{cov_estimates_K} will become negligible.
\end{remark}

In this paper, we deal with the problem of estimating from samples of a stationary process its generalized cepstral coefficients. The latter object is obtained via replacing the logarithm in the definition of the classical cepstral coefficients \eqref{cep_def} with the so-called \emph{generalized logarithmic function}:
\begin{equation}
s_{\alpha}(x) = 
\begin{cases}
(x^\alpha-1)/\alpha, & 0<|\alpha|\leq 1, \\
\log x, & \alpha=0.
\end{cases}
\end{equation}
The generality comes from the fact that for a fixed $x>0$, the limit $(x^\alpha-1)/\alpha \to \log x$ holds as the parameter $\alpha\to 0$. The generalized logarithm leads to the following definition where we are mostly interested in positive $\alpha$'s.
\begin{definition}[\hspace{1sp}\cite{tokuda1990generalized}]
	Given a power spectral density $\Phi(\theta)$, a real number $\alpha\in (0,1]$, the generalized cepstral coefficients are defined as 
	\begin{equation}\label{def_cep_gen}
	m_{\alpha,k} = \begin{cases}
	\frac{1}{\alpha} \int_{\Tbb} e^{i k \theta} \Phi(\theta)^\alpha \frac{\d\theta}{2\pi} & \text{if } k\neq 0, \\
	 \\
	\frac{1}{\alpha} \left[ \int_{\Tbb} \Phi(\theta)^\alpha \frac{\d\theta}{2\pi} -1 \right] & \text{if } k=0.
	\end{cases}
	\end{equation}
\end{definition}
Notice that the constant one in the numerator of $(x^\alpha-1)/\alpha$ disappears in the case of $k\neq 0$ due to the integration against the complex exponential function $e^{ik\theta}$.  Moreover, we are not so interested in $\alpha=1$ since in that case, $m_{\alpha,k}$ reduces to the covariance $c_k$ in \eqref{cov_def} (with a difference of one for $k= 0$).
Thus in the remaining part of this paper, we only consider the case $\alpha\in(0,1)$.
With a slight abuse of the notation, we shall omit $\alpha$ in the subscript and simply write $m_k$ instead.

%Inspired by \eqref{cov_estimates}, an estimator for $m_k$ is immediately available as
%\begin{equation}\label{cep_estimates}
%\hat{m}_k = \frac{1}{N} \sum_{\ell=0}^{N-1} e^{ik\frac{2\pi}{N}\ell} \log\hat{\Phi}_\ell,
%\end{equation}
%which is obtained from \eqref{cep_def} by replacing $\Phi$ with $\hat{\Phi}$ and discretizing the integral on the grid \eqref{T_grid}.

An estimator of the generalized cepstral coefficients using the periodogram, inspired by \eqref{cov_estimates}, is given as
\begin{equation}\label{est_cep_gen}
\hat{m}_k = \begin{cases}
\frac{1}{\alpha N} \sum_{\ell=0}^{N-1} e^{ik\frac{2\pi}{N}\ell} \hat{\Phi}_\ell^\alpha & \text{if } k\neq 0, \\
 \\
\frac{1}{\alpha} \left( \frac{1}{N} \sum_{\ell=0}^{N-1} \hat{\Phi}_\ell^\alpha -1 \right) & \text{if } k=0.
\end{cases}
\end{equation}
We end this section by formally stating the consistency problem for the generalized cepstral estimator above.

\begin{problem}
	Given a sample path $y_0, \dots, y_{N-1}$ of a zero-mean stationary complex random process $y_t$, {understand} whether the estimator \eqref{est_cep_gen} converges to the true value \eqref{def_cep_gen} in some stochastic sense as the length of the sample path $N$ tends to infinity.
\end{problem}

%DO WE care about the index $k=0$? BETTER include this for the sake of completeness...

\section{Consistency in the Case of Independent Spectral Components}\label{sec:independ}

First, let us introduce a normalized version of the spectral components:
\begin{equation}
\bar{Y}_\ell := \frac{1}{\sqrt{N}} Y_\ell \ \text{ so that }\ \hat\Phi_\ell=|\bar{Y}_\ell|^2.
\end{equation}
%To simplify the presentation, from now on we shall take the number of spectral components $K=2N-1$ which is the smallest number satisfying the condition before \eqref{cov_estimates}.
%In this case, the spectral components $\{\bar Y_\ell\}_{\ell=0}^{N-1}$ corresponds precisely to the frequency range $[0,\pi)$.
%Due to the symmetry of the spectrum $\hat\Phi(\pi-\theta) = \hat{\Phi}(\pi+\theta)$ for a real process $y$, we have the relations
%\begin{equation}\label{spec_symmetry}
%\hat{\Phi}_{N-1} = \hat{\Phi}_N, \ \hat{\Phi}_{N-2} = \hat{\Phi}_{N+1},\ \dots,\  \hat{\Phi}_1 = \hat{\Phi}_{K-1}.
%\end{equation}
%In other words, the discrete spectral content is concentrated in $\ell\in\{0,\dots,N-1\}$.
%Furthermore, 
%The Gaussian assumption for our signal $y_t$ is stated next.

\begin{assumption}\label{assump_independ}
	The normalized spectral components $\{\bar Y_\ell\}_{\ell=0}^{N-1}$ of the signal $\{y_t\}_{t=0}^{N-1}$ are independent zero-mean complex (circular) Gaussian random variables such that $\var \bar{Y}_\ell = \E\, \hat{\Phi}_\ell = \lambda_\ell>0$. 
	%More specifically, since we have taken an odd $K=2N-1$,
	%[\bin{DETERMINE whether we want $K=N$ or not. We can also choose $N$ even for the presentation.}],
	%\begin{itemize}
%		\item $Y_0$ is real Gaussian;
		%\item for other $\ell\in\{1,\dots,N-1\}$, 
	More precisely, the real and imaginary parts of each $Y_\ell$ are independent real Gaussian random variables having equal variance, namely $\lambda_\ell/2$, see e.g., \cite[p.~361]{stoica2005spectral}.
	%\end{itemize}
\end{assumption}

%\begin{remark}
%	The choice $K=2N-1$ proves to lead to the simplest presentation for our generalized cepstral estimation theory. One can in principal choose larger values of $K$, in which case however, the independence assumption above for the spectral components residing in the half interval $[0,\pi)$ must fail because one has only $N$ time-domain samples. The proof has to be adapted taking into consideration this subtle point.
%\end{remark}

Observe that the above assumption holds in the following two cases.

\emph{Case 1.} The process $y_t$ is i.i.d. Gaussian with a variance $\sigma^2>0$, e.g., a Gaussian white noise. In this case, the random vector $\bar{\Yb}=\frac{1}{\sqrt{N}}\Fb\yb$ (see \eqref{vec_notation} for the notation) is Gaussian with a covariance matrix
\begin{equation}
%\begin{aligned}
\E \bar{\Yb}\bar{\Yb}^*  = \frac{1}{N}\Fb (\E\yb\yb^*) \Fb^* = \sigma^2 \Ib,
%\end{aligned}
\end{equation}
where $\Fb$ is the DFT matrix in \eqref{DFT_mat}.

\emph{Case 2.} The process $y_t$ is Gaussian and $N$-periodic, i.e., $y_t=y_{t+N}$ almost surely for all $t\in\Zbb$. An immediate consequence is the cyclic symmetry $c_{N-k}=c_{-k}=c_k^*$ for the covariances. It follows that the covariance matrix of the random vector $\yb$ has a \emph{circulant} (more than being Toeplitz) structure, namely
%WRITE OUT the circulant matrix $\Cb$
\begin{equation}
\Cb := \E \yb\yb^* 
= \bmat c_0 & c_{N-1} & \cdots & c_2 & c_1 \\
c_1 & c_0 & c_{N-1} &  & c_2 \\
\vdots & c_1 & c_0 & \ddots & \vdots \\
c_{N-2} &  & \ddots & \ddots & c_{N-1} \\
c_{N-1} & c_{N-2} & \cdots & c_1 & c_0 \emat.
\end{equation}
In plain words, each column of $\Cb$ is the cyclic shift of the first column (the same for rows). Since $\Cb$ is also a covariance matrix, we shall in addition require it to be positive definite\footnote{The Hermitian structure of $\Cb$ comes from the relation $c_{N-k}=c_k^*$.}. It is well known that any circulant matrix can be diagonalized by the DFT matrix. More precisely, define a unitary matrix $\Ub := \frac{1}{\sqrt{N}}\Fb$ and let $\cb = [c_0,\dots,c_{N-1}]^\top$ be the first column of $\Cb$. Then the cirulant matrix $\Cb$ admits a {spectral} decomposition $\Cb=\Ub^* \Psi \Ub$ where the columns of {$\Ub^*$} are the normalized eigenvectors and the diagonal matrix $\Psi = \diag(\Fb\cb)$ contains the positive eigenvalues \cite{Gray-2006}, see also \cite{LPcirculant-13}.
Therefore, the covariance matrix of the vector $\bar{\Yb}$ of spectral components is
%[CHECK the formula below...]
%\begin{equation}
%\bar{\Yb} = \frac{1}{\sqrt{N}} \verb|fft| ( \yb )
%\end{equation}
%where $\bar{\Yb} = [\bar{Y}_0,\dots,\bar{Y}_{N-1}]$ and only the first $N$ points are of interest in the \emph{$K$-point} FFT.
%[EXPAND this part, but briefly...]
\begin{equation}
%\begin{aligned}
\E \bar{\Yb}\bar{\Yb}^* = \frac{1}{N}\Fb \Cb \Fb^* = \Psi.
%\end{aligned}
\end{equation}

%$\Fb\Ub^* = \sqrt{N} \Ib$

%Both cases can be understood from the fact that circulant matrices can be diagonalized by the DFT matrix which is unitary up to a constant factor $N$, see \cite{Gray-2006} and refer to Sec.~\ref{sec:review_DFT}.

The main result of this section is summarized below.

\begin{theorem}\label{thm_gen_cep_est}
	Under Assumption~\ref{assump_independ}, the estimator \eqref{est_cep_gen} of the generalized cepstral coefficients is consistent up to a constant multiplicative factor. More precisely, we have for $k\neq 0$,
	\begin{equation}\label{ms_converg_hat_mk}
	C \hat{m}_k = \frac{C}{\alpha N} \sum_{\ell=0}^{N-1} e^{ik\frac{2\pi}{N}\ell} \hat{\Phi}_\ell^\alpha \xrightarrow{\text{m.s.}} m_k
	\end{equation}
	and
	\begin{equation}\label{ms_converg_hat_m0}
	C\hat{m}_0 +\frac{C-1}{\alpha} = \frac{C}{\alpha N} \sum_{\ell=0}^{N-1} \hat{\Phi}_\ell^\alpha - \frac{1}{\alpha} \xrightarrow{\text{m.s.}} m_0
	\end{equation}
	as $N\to\infty$, where the constant $C = 1/\Gamma(\alpha+1)$ and the convergence is understood in the mean-square sense. Here $\Gamma(\cdot)$ is the gamma function.
\end{theorem}

%\begin{lemma}\label{lem_ms_converg}
%	Let $X_n$ be a sequence of complex-valued random variables such that $\E X_n \to C$ as $n\to\infty$
%	%	\begin{equation}\label{lim_mean_var}
%	%	\E X_k \to C \ \text{ and }\ \var X_k \to 0 \ \text{ as }\ k\to\infty, 
%	%	\end{equation}
%	where $C$ is a constant. Then we have $X_n\xrightarrow{\text{m.s.}} C$ if and only if $\var X_n \to 0$ as $n\to\infty$.
%\end{lemma}
%
%\bin{[REMOVE this Lemma and Proof and write an in-text sentence instead.]}
%
%\begin{proof}
%	For the ``if'' part, notice first that the second moment 
%	\begin{equation}
%	\E (|X_n|^2) = |\E X_n|^2 + \var X_n \to |C|^2 \ \text{ as }\ n\to\infty.
%	\end{equation}
%	%is bounded because so are both terms on the right-hand side due to the limit relation \eqref{lim_mean_var}.
%	Next, we simply check the definition of mean-square convergence:
%	\begin{equation}
%	\E (|X_n-C|^2) = \E(|X_n|^2 -2\Re(C\bar{X}_n) +|C|^2) \to 0
%	\end{equation}
%	as $n\to\infty$.
%	The ``only if'' part can be proved by reversing the above argument.
%\end{proof}

%In light of Lemma~\ref{lem_ms_converg},

The proof of this theorem is built on a well-known fact: if (and only if) the variance of a sequence of complex-valued random variables tends to zero, then the sequence itself converges to a constant in the mean-square sense. Moreover, such a constant is the limit of the expectation of the sequence.
Therefore, the claim of Theorem~\ref{thm_gen_cep_est} can be established if we can show in the case of $k\neq 0$ that $\E(C\hat{m}_k) \to m_k$ for the given constant $C$ and that $\var \hat{m}_k \to 0$ as the number of samples $N$ tends to infinity, and similarly for the case $k=0$.
The rest of this section is devoted to the computation of the expectation and the variance of $\hat{m}_k$.
%The following integral formula from \cite[Eq.~3.381.4]{gradshteyn2014table} will be used several times \bin{[SEVERAL?]} in the derivation, and we include it here for future reference:
%\begin{equation}\label{integral_formula}
%\int_{0}^{\infty} x^{\nu-1} e^{-\mu x} \d x = \frac{1}{\mu^{\nu}} \Gamma(\nu),\quad \Re\,\mu>0,\ \Re\,\nu>0.
%\end{equation}

\subsection{Mean of $\hat{\Phi}_\ell^\alpha$ and $\hat{m}_k$}

%[MAYBE not necessary as a subsection...]

%\emph{First case}: $\ell\in\{1,\dots,N-1\}$ so that each $\bar{Y}_\ell$ is complex Gaussian. 
It follows from Assumption~\ref{assump_independ} that $\hat\Phi_\ell=|\bar{Y}_\ell|^2 = (\Re \bar{Y}_\ell)^2 + (\Im \bar{Y}_\ell)^2$ is $\chi^2$-distributed with a degree of freedom $2$ after suitable scaling.
%[MAYBE too trivial to write out...] $\frac{2}{\lambda_\ell} \hat{\Phi}_\ell \sim \chi^2(2)$, and 
In fact, the probability density function of $\hat{\Phi}_\ell$ is
\begin{equation}
f(x) = \frac{1}{\lambda_\ell} e^{-x/\lambda_\ell}, \quad x\geq 0,
\end{equation}
which is just an exponential distribution.
The following computation is straightforward:
\begin{equation}\label{expect_hatPhi_power}
\E(\hat{\Phi}_\ell^\alpha) = \frac{1}{\lambda_\ell} \int_{0}^{\infty} x^{\alpha} e^{-x/\lambda_\ell} \d x = \lambda_\ell^\alpha \Gamma(\alpha+1),
\end{equation}
where the second equality comes from \cite[Eq.~3.381.4]{gradshteyn2014table}. % \eqref{integral_formula}.

%For $\ell\in\{N,\dots,K-1\}$, we have by \eqref{spec_symmetry} that
%\begin{equation}
%\E(\hat{\Phi}_\ell^\alpha) = \E(\hat{\Phi}_{K-\ell}^\alpha) = \lambda_{K-\ell}^\alpha \,\Gamma(\alpha+1).
%\end{equation}

%\emph{Second case}: $\ell=0$ so that $\bar{Y}_0$ is real Gaussian. Now $\hat{\Phi}_0 = |\bar{Y}_0|^2$ follows the $\chi^2(1)$ distribution after suitable scaling. More precisely, $\hat{\Phi}_0$ has a probability density function
%\begin{equation}
%f(x) = \frac{1}{\sqrt{2\pi\lambda_0}} x^{-\frac{1}{2}} e^{-\frac{x}{2\lambda_0}},\quad x>0.
%\end{equation}
%We proceed to compute
%\begin{equation}\label{expect_hatPhi_zero_power}
%\begin{aligned}
%\E (\hat{\Phi}_0^\alpha) & = \frac{1}{\sqrt{2\pi\lambda_0}} \int_{0}^{\infty}  x^{\alpha-\frac{1}{2}} e^{-\frac{x}{2\lambda_0}} \d x \\
% & = \frac{1}{\sqrt{\pi}} (2\lambda_0)^\alpha \,\Gamma(\alpha+1/2),
%\end{aligned}
%\end{equation}
%where we have used again Formula \eqref{integral_formula}.

We are now ready to compute the mean of $\hat{m}_k$.
%\emph{First case.} 
For $k\neq 0$, we have
\begin{subequations}\label{expect_hat_mk}
\begin{align}
\E \,\hat{m}_k & = \E \left( \frac{1}{\alpha N} \sum_{\ell=0}^{N-1} e^{ik\frac{2\pi}{N}\ell} \hat{\Phi}_\ell^\alpha \right) \nonumber \\
 & = \frac{1}{\alpha N} \sum_{\ell=0}^{N-1} e^{ik\frac{2\pi}{N}\ell} \lambda_\ell^\alpha \Gamma(\alpha+1) \label{expect_hat_mk_a} \\
 & \to \frac{\Gamma(\alpha+1)}{\alpha} \int_{\Tbb} e^{i k \theta} \Phi(\theta)^\alpha \frac{\d\theta}{2\pi} \quad \text{ as }\ N\to\infty. \label{expect_hat_mk_b}
\end{align}
\end{subequations}
To see the latter limit, recall first that under a mild condition for the decay of the true covariances $c_k$ of the process $y_t$, $\E\, \hat{\Phi}_\ell = \lambda_\ell \to \Phi(2\pi\ell/N)$ as $N\to\infty$ for each $\ell\in\{0,\dots,N-1\}$ where $\Phi$ is the true spectrum \cite[p.~7]{stoica2005spectral}.
%Moreover, from the previous computations, we have for $\ell\in\{0,\dots,N-1\}$, $\E \hat{\Phi}_\ell=\lambda_\ell$; for $\ell\in\{N,\dots,K-1\}$, $\E \hat{\Phi}_\ell = \lambda_{K-\ell}$.
Therefore, the term $\frac{1}{N} \sum_{\ell=0}^{N-1} e^{ik\frac{2\pi}{N}\ell} \lambda_\ell^\alpha$ in \eqref{expect_hat_mk_a} converges to the integral in \eqref{expect_hat_mk_b}, where the convergence is understood from the normalized Riemann sum to the integral. 
%The last two terms in the bracket in \eqref{expect_hat_mk_a} are bounded, so the product of them and $\frac{1}{\alpha K}$ tends to zero as $N$ and $K$ tends to infinity.

%\bin{[CHECK whether this result holds for real processes, and check whether the decay condition is compatible with the Gaussian assumption (BOTH should be OK)...]}

Similarly for $k=0$, we have %$\E \,\hat{m}_0 = \E [ \frac{1}{\alpha} ( \frac{1}{N} \sum_{\ell=0}^{N-1} \hat{\Phi}_\ell^\alpha -1 ) ]$
\begin{equation}
\begin{aligned}\label{expect_hat_m0}
\E \,\hat{m}_0 + \frac{1}{\alpha} & = \E \left( \frac{1}{\alpha N} \sum_{\ell=0}^{N-1} \hat{\Phi}_\ell^\alpha \right) \\
 & \to \frac{\Gamma(\alpha+1)}{\alpha} \int_{\Tbb} \Phi(\theta)^\alpha \frac{\d\theta}{2\pi} \quad \text{ as }\ N\to\infty,
\end{aligned}
\end{equation}
%$C(\E\,\hat{m}_0 +1/\alpha)-1/\alpha$
which is equivalent to $\E [C\hat{m}_0 + (C-1)/\alpha]\to m_0$ for the random variable in \eqref{ms_converg_hat_m0} and the constant $C$ given in Theorem~\ref{thm_gen_cep_est}.

%[COMPUTE directly the expectation and the second moment...]

\subsection{Variance of $\hat{\Phi}_\ell^\alpha$ and $\hat{m}_k$}

%[MAYBE not necessary as a subsection...]

%\emph{First case}: $\ell\in\{1,\dots,N-1\}$. 
The second moment of $\hat{\Phi}_\ell^\alpha$ is just
\begin{equation}
\E [(\hat{\Phi}_\ell^\alpha)^2] = \E (\hat{\Phi}_\ell^{2\alpha}) = \lambda_\ell^{2\alpha} \Gamma(2\alpha+1),
\end{equation}
where the second equality follows immediately from \eqref{expect_hatPhi_power} since they have the same functional form.
Consequently, we have
\begin{equation}\label{var_hatPhi_l}
\begin{aligned}
\var (\hat{\Phi}_\ell^\alpha) & = \E (\hat{\Phi}_\ell^{2\alpha}) - [\E(\hat{\Phi}_\ell^\alpha)]^2 \\
 & = \lambda_\ell^{2\alpha} \left\{ \Gamma(2\alpha+1) - [\Gamma(\alpha+1)]^2 \right\}.
\end{aligned}
\end{equation}
%Similar to the computation in the previous subsection, for $\ell\in \{N,\dots,K-1\}$, we have
%\begin{equation}
%\var (\hat{\Phi}_\ell^\alpha)  = \var (\hat{\Phi}_{K-\ell}^\alpha) = \lambda_{K-\ell}^{2\alpha} \left\{ \Gamma(2\alpha+1) - [\Gamma(\alpha+1)]^2 \right\}.
%\end{equation}

%\emph{Second case}: $\ell=0$. The second moment of $\hat{\Phi}_0^\alpha$ can be obtained from \eqref{expect_hatPhi_zero_power}:
%\begin{equation}
%\E (\hat{\Phi}_0^{2\alpha}) = \frac{1}{\sqrt{\pi}} (2\lambda_0)^{2\alpha} \,\Gamma(2\alpha+1/2).
%\end{equation}
%It follows that
%\begin{equation}\label{var_hatPhi_0}
%\var (\hat{\Phi}_0^\alpha) = \frac{(2\lambda_0)^{2\alpha}}{\sqrt{\pi}}  \left\{ \Gamma(2\alpha+1/2) - \frac{1}{\sqrt{\pi}} [\Gamma(\alpha+1/2)]^2 \right\}.
%\end{equation}

\begin{example}\label{ex_alpha_nu}
	In the papers \cite{Zhu-Zorzi-21,Zhu-Zorzi-2021-cepstral}, we took some special values of the parameter $\alpha$, i.e.,
	\begin{equation}\label{cond_alpha}
	\alpha = 1-\frac{1}{\nu},\quad \nu\in\Nbb_+, \quad \nu\geq 2.
	\end{equation}
	For instance, take $\nu=2$ so $\alpha=1/2$,
	% and $\ell\in\{0,\dots,K-1\}\backslash\{0,K/2\}$, 
	and we have
%	\begin{itemize}
%		\item for $\ell\in\{1,\dots,N-1\}$,
		\begin{equation}
		\begin{aligned}
		\E(\hat{\Phi}_\ell^\alpha) & = \sqrt{\lambda_\ell}\,\Gamma(3/2) = \frac{\sqrt{\pi \lambda_\ell}}{2}, \\  
		\E(\hat{\Phi}_\ell^{2\alpha}) & = \lambda_\ell \Gamma(2) = \lambda_\ell,
		\end{aligned}
		\end{equation}
		and $\var(\hat{\Phi}_\ell^\alpha) = (1-\pi/4) \lambda_\ell$.
%		\item for $\ell=0$,
%		\begin{equation}
%		\begin{aligned}
%		\E(\hat{\Phi}_0^\alpha) & = \sqrt{\frac{2\lambda_0}{\pi}}, \\  
%		\E(\hat{\Phi}_0^{2\alpha}) & = \frac{2\lambda_0}{\sqrt{\pi}} \Gamma(3/2) = \lambda_0,
%		\end{aligned}
%		\end{equation}
%		and $\var(\hat{\Phi}_0^\alpha) = (1-2/\pi) \lambda_0$.
%	\end{itemize}
\end{example}

Now we can continue to compute the variance of $\hat{m}_k$, where the independence between the spectral components in Assumption~\ref{assump_independ} plays an important role.
For $k\neq 0$, we have 
\begin{subequations}\label{calc_var_hat_mk}
\begin{align}
%\var(\hat{m}_k) & := \E \, ( |\hat{m}_k - \E\hat{m}_k|^2 ) \nonumber \\
% & = \E \left\{ \left| \frac{1}{\alpha N} \sum_{\ell=0}^{N-1} e^{ik\frac{2\pi}{N}\ell} \left[ \hat{\Phi}_\ell^\alpha -  \E(\hat{\Phi}_\ell^\alpha) \right] \right|^2 \right\} \nonumber \\
%= & \frac{1}{\alpha^2 N^2} \E \left\{ \left[ (\hat{\Phi}_0^\alpha - \E \hat{\Phi}_0^\alpha) \right. \right. \nonumber \\
% & \qquad \left. \left. + \sum_{\ell=1}^{N-1} 2\cos(k\frac{2\pi}{K}\ell) (\hat{\Phi}_\ell^\alpha - \E \hat{\Phi}_\ell^\alpha) \right]^2 \right\} \label{calc_var_step3} \\
\var(\hat{m}_k) & = \var \left( \frac{1}{\alpha N} \sum_{\ell=0}^{N-1} e^{ik\frac{2\pi}{N}\ell} \hat{\Phi}_\ell^\alpha \right) \nonumber \\
 & = \frac{1}{\alpha^2 N^2} \sum_{\ell=0}^{N-1} \var (\hat{\Phi}_\ell^\alpha) \label{calc_var_step4} \\
 & = \frac{C_1}{\alpha^2 N^2}  \sum_{\ell=0}^{N-1} \lambda_\ell^{2\alpha} \to 0 \ \text{ as }\ N\to\infty, \label{calc_var_step5}
% &  \nonumber
\end{align}
\end{subequations}
where,
%\begin{itemize}
%	\item in \eqref{calc_var_step3}, we have used \eqref{spec_symmetry};
%	\item 
	in \eqref{calc_var_step4} we have used the independence assumption between the spectral components $\{\bar Y_\ell\}_{\ell=0}^{N-1}$, and in \eqref{calc_var_step5} the constant $C_1 = \Gamma(2\alpha+1) - [\Gamma(\alpha+1)]^2$
%	\begin{equation}
%	\begin{aligned}
%	C_1 & = \Gamma(2\alpha+1) - [\Gamma(\alpha+1)]^2, \\
%	C_2 & = \frac{2^{2\alpha}}{\sqrt{\pi}} \left\{ \Gamma(2\alpha+1/2) - \frac{1}{\sqrt{\pi}} [\Gamma(\alpha+1/2)]^2 \right\}
%	\end{aligned}
%	\end{equation}
	is determined from \eqref{var_hatPhi_l}. % and \eqref{var_hatPhi_0}.
%\end{itemize}
To see the last limit, 
%notice that the first term in the bracket in \eqref{calc_var_step5} is bounded as explained after \eqref{expect_hat_mk_b} so its product with $\frac{1}{\alpha^2 K^2}$ tends to zero as $K\to\infty$. 
%For the second term in \eqref{calc_var_step5}, 
recognize
\begin{equation}\label{converg_var_hat_mk}
\frac{1}{N} \sum_{\ell=0}^{N-1} \lambda_\ell^{2\alpha} \to \int_{0}^{2\pi} \Phi(\theta)^{2\alpha} \frac{\d\theta}{2\pi}
\end{equation}
as the convergence of the (normalized) Riemann sum to the integral. Hence the Riemann sum is bounded, and its product with $\frac{C_1}{\alpha^2 N}$ tends also to zero as $N\to\infty$.

For $k=0$, we have 
\begin{equation}
%\begin{aligned}
\var (\hat{m}_0) = \var  \left( \frac{1}{\alpha N} \sum_{\ell=0}^{N-1} \hat{\Phi}_\ell^\alpha \right) 
% & = \frac{1}{\alpha^2 N^2} \sum_{\ell=0}^{N-1} \var (\hat{\Phi}_\ell^\alpha)
%\end{aligned}
\end{equation}
which leads to the same result as \eqref{calc_var_step4}. Therefore, we also have $\var (\hat{m}_0)\to0$ as $N\to \infty$.
%{\mz [MATTIA: It seems to me that for the variance there is no need to split in the cases $k=0$ and $k\neq 0$. Is that correct? \bin{Bin: I wrote (35) because there is a constant $1$ in $\hat{m}_0$, see \eqref{est_cep_gen}. Anyway, the case for $k=0$ is only mentioned very briefly.}] }

%\emph{Second case}: $\ell\in\{0,K/2\}$.

\begin{remark}\label{rem_regular_Phi}
	It is implicitly assumed that the underlying true spectral density function $\Phi$ is sufficiently ``nice'' in order for such convergences as \eqref{expect_hat_mk_b} and \eqref{converg_var_hat_mk} to make sense. For example, when the true covariance sequence $c_k$ belongs to the Wiener class, i.e., $\sum_{k=-\infty}^{\infty} |c_k|<\infty$, the corresponding spectral density $\Phi$ is in fact continuous (a consequence of Lebesgue's dominated convergence theorem).
\end{remark}

%[MENTION the functional space for $\Phi$...]
\begin{remark}\label{rem_Lp_Phi}
	In view of the formulas \eqref{cov_def}, \eqref{expect_hat_mk_b}, and \eqref{converg_var_hat_mk}, the true spectral density $\Phi$ of the underlying process $y$ has to satisfy the integrability condition: $\Phi\in L^1(\Tbb)\cap L^\alpha(\Tbb)\cap L^{2\alpha}(\Tbb)$. Since the set $\Tbb$ has a finite Lebesgue measure, we have for $0<p<q\leq\infty$ that $L^q(\Tbb)\subset L^p(\Tbb)$ \cite{Villani-85}. Recall that the parameter $\alpha$ takes value in the interval $(0,1)$. Therefore, the former intersection reduces to $L^1(\Tbb)$ if $0<\alpha\leq 1/2$, and to $L^{2\alpha}(\Tbb)$ if $1/2<\alpha<1$. For both cases, $\Phi\in L^2(\Tbb)$ is a sufficient condition.
\end{remark}

\section{The Case of Correlated Spectral Components}\label{sec:depend}

%{\mz 
	
Uncorrelatedness between spectral components as dictated by Assumption~\ref{assump_independ} holds in an asymptotic sense \cite[Theorem~6.2.3, p.~426]{priestley1981spectral}. With a finite number of measurements (samples) however, two spectral components $\bar Y_{\ell_1}$ and $\bar Y_{\ell_2}$ with $\ell_1\neq \ell_2$ are in general correlated. 
Indeed, a direct computation using the unnormalized spectral components \eqref{spec_comp_theta} leads to the expression
\begin{equation}
\cov(Y_{\ell_1}, Y_{\ell_2}) = \frac{1}{1- e^{-i(\theta_1-\theta_2)}} \left[ \sum_{k=0}^{N-1} c_k a_k - \sum_{k=-N+1}^{-1} c_k a_k  \right]
\end{equation}
where $\theta_j = 2\pi\ell_j/N$ for $j=1, 2$, $c_k$ {is} the true covariance of the process $y_t$ in \eqref{cov_def}, and $a_k = e^{-i\theta_1 k} - e^{-i\theta_2 k}$.

%\bin{[CHECK AND ADAPT Eq.~(3) in \cite{ephraim2005second} to complex processes, and then state the correlatedness argument.]}
	
The consistency problem for the classical cepstral estimation in the case of correlated spectral components has been investigated in  \cite{ephraim2005second}. In what follows we review some well known facts that will be used later.  Assume that $\bar Y_{\ell_1}$ and $\bar Y_{\ell_2}$ are zero-mean jointly {circular} Gaussian random variables, that is, we drop the independence part in Assumption~\ref{assump_independ}.
%{\mz [IT SEEMS TO ME THAT THIS HYPOTHESIS ALWAYS HOLD IN OUR GAUSSIAN SETTING, RIGHT? \bin{Bin: Sure, if the time series $y_t$ are zero-mean jointly Gaussian, then so are the spectral components $\bar{Y}_{\ell}$ because they are obtained via a linear (Fourier) transform of $y_t$.}]}
%{\mz [It is well known that the real and the imaginary components of the random vector $[\, \bar Y_{\ell_1}\; \bar Y_{\ell_2}\,]^\top$ are uncorrelated and have the same variance, see for instance \cite{whalen2013detection}. \bin{Bin: I would remove this sentence since it has the same meaning as the word ``circular Gaussian'' as stated also in Assumption~\ref{assump_independ}.} ]}
Let 
\begin{equation}
R_{\ell_1\ell_2}=\left[\begin{array}{cc} \lambda_{\ell_1} & r_{\ell_1\ell_2} \\r_{\ell_1\ell_2}^* & \lambda_{\ell_2}  \\\end{array}\right]
\end{equation}
be the covariance matrix of the random vector $[\, \bar Y_{\ell_1}\; \bar Y_{\ell_2}\,]^\top$. According to \cite[Eq.~16]{ephraim2005second}, we have for $\alpha\in(0,1)$ that
\begin{align}\label{cross_exp}
\E[\hat{\Phi}_{\ell_1}^\alpha
\hat{\Phi}_{\ell_2}^{\alpha}] & =\E[\,|\bar Y_{\ell_1}|^{2\alpha}|\bar Y_{\ell_2}|^{2\alpha}\,]\nn \\ 
 & =\lambda_{\ell_1}^\alpha \lambda_{\ell_2}^\alpha [\Gamma(\alpha+1)]^2 \,{}_2F_1(-\alpha,-\alpha;1;\rho_{\ell_1\ell_2}^2)
\end{align}
%\bin{[Bin: I have simplified this part because I see the separate discussions for $\theta\in(0, \pi)$ and $\theta\in\{0, \pi\}$ in \cite{ephraim1999second,ephraim2005second} as a nuisance due to the assumption that the underlying time series is \emph{real}. In that case, indeed the Fourier transform $Y(\theta)$ in \eqref{spec_comp_theta} are real-valued for $\theta=0$ or $\pi$. This fact changes the computations
%\begin{itemize}
%	\item in (26) since now $\bar{Y}_\ell$ is real Gaussian and $\hat{\Phi}_\ell$ follows a $\chi^2(1)$ distribution,
%	\item and in \cite[Eq.~11]{ephraim2005second} since now the conditional density \cite[Eq.~10]{ephraim2005second} changes.
%\end{itemize}	
%Such a nuisance can be resolved by simply assuming that the underlying time series $y_t$ is \emph{complex} so that we can safely work only in one case, as I have done in previous sections.]
%}
where,%\footnote{This definition is slightly different from the one in \cite{ephraim2005second} in order to ease the notation.}  
$$\rho_{\ell_1\ell_2}^2 := |r_{\ell_1\ell_2}|^2/(\lambda_{\ell_1}\lambda_{\ell_2})
%\left\{\begin{array}{ll} 2|r_{\ell_1\ell_2}|^2/(\lambda_{\ell_1}\lambda_{\ell_2}), & \ell_1\in\mathcal S \hbox{ and }  \ell_2\in\bar{\mathcal S}\\
% 2|r_{\ell_1\ell_2}|^2/(\lambda_{\ell_1}\lambda_{\ell_2}), & \ell_1\in\bar{\mathcal S} \hbox{ and }  \ell_2\in\mathcal S \\ |r_{\ell_1\ell_2}|^2/(\lambda_{\ell_1}\lambda_{\ell_2}), & \hbox{otherwise}
%  \end{array}\right.
$$
is the modulus squared of the Pearson correlation coefficient between $\bar Y_{\ell_1}$ and $\bar Y_{\ell_2}$, and
%with $\mathcal S=\{\ell \hbox{ s.t. } \ell /N=0,0.5\}$,
%$\bar{\mathcal S}=\{\ell \hbox{ s.t. } 0<\ell /N<0.5, \,0.5< \ell<1\}$; 
%$$
\begin{equation}\label{hypergeom}
{}_2 F_1(a,b;c;z)=\sum_{n=0}^\infty \frac{(a)_n (b)_n}{(c)_n}\frac{1}{n!}z^n
\end{equation} 
%$$
is the hypergeometric series \cite[Ch.~15]{abramowitz1972handbook} with a convergence region $|z|<1$ (when $c$ is not a negative integer). Here
$$ (a)_n=\left\{\begin{array}{ll}1, & n=0  \\ a(a+1)\cdot \ldots \cdot (a+n-1), & n>0  \end{array}\right. $$ 
is the \emph{Pochhammer symbol} for the \emph{rising factorial}.
In addition, the series \eqref{hypergeom} also converges on the unit circle $|z|=1$ if $\Re(c-a-b)>0$ which holds in the case of \eqref{cross_exp}.
% in the case that $-a,-b,c$ are positive real \bin{[CHANGE here to $|z|<1$ and the convergence condition for $|z|=1$.]} and
%$$ C_{\ell_1 \ell_2}=\left\{\begin{array}{ll}\pi^{-1}2^{2\alpha}\Gamma(\alpha+0.5)^2, & \ell_1,\ell_2\in\mathcal S\\ 
%\sqrt{\pi}^{-1}2^{\alpha}\Gamma(\alpha+0.5)\Gamma(\alpha+1), & \ell_1\in\bar{\mathcal S} \hbox{ and }  \ell_2\in\mathcal S\\ 
%\sqrt{\pi}^{-1}2^{\alpha}\Gamma(\alpha+0.5)\Gamma(\alpha+1), & \ell_1\in\mathcal S \hbox{ and }  \ell_2\in\bar{\mathcal S}\\ 
%  \Gamma(\alpha+1)^2, & \hbox{otherwise}; \end{array}\right. $$
%$$ \beta_{\ell_1 \ell_2}=\left\{\begin{array}{ll}0.5, & \ell_1,\ell_2\in\mathcal S\\ 
%1, & \hbox{otherwise}. \end{array}\right. $$
%It is worth noting that $C_{\ell_1\ell_2}>0$ and finite for any $\ell_1,\ell_2$; let $C$ denote its maximum with respect to $\ell_1$ and $\ell_2$.   

%We are now ready to prove that the estimator proposed in the previous section is consistent also in this case, i.e. that  Theorem \ref{thm_gen_cep_est} still holds provided that  $\rho_{\ell_1\ell_2}^2$ obeys to a certain property which will be introduced later.
Next, we show that in the case with correlated spectral components, the statistical consistency of the estimator \eqref{est_cep_gen} using the unwindowed periodogram \emph{cannot} be guaranteed in general, although (\ref{expect_hat_mk}) and (\ref{expect_hat_m0}) still hold because they have been derived without using the independence assumption. %It remains to prove that $\var(\hat{m}_k)\rightarrow 0$ as $k\rightarrow \infty$. 
We shall only work on $\hat{m}_k$ with $k\neq 0$ since the conclusion for the case of $k=0$ follows similarly.

%Indeed, i
It is not difficult to see that
\begin{align}\label{var_hat_m}
 & \var(\hat m_k)=\E\left[ |\hat m_k - \E\,\hat m_k|^2 \right] \nn\\
= & \frac{1}{\alpha^2 N^2} \sum_{\ell_1=0}^{N-1} \sum_{\ell_2=0}^{N-1} e^{ik\frac{2\pi}{N}(\ell_1-\ell_2)} \left[ \E(\hat{\Phi}_{\ell_1}^\alpha
\hat{\Phi}_{\ell_2}^{\alpha})-\E(\hat{\Phi}_{\ell_1}^\alpha) \E(\hat{\Phi}_{\ell_2}^{\alpha}) \right] \nn\\
= & \frac{1}{\alpha^2 N^2} \sum_{\ell_1=0}^{N-1} \sum_{\ell_2=0}^{N-1} e^{ik\frac{2\pi}{N}(\ell_1-\ell_2)} \lambda_{\ell_1}^\alpha \lambda_{\ell_2}^\alpha [\Gamma(\alpha+1)]^2 \nn\\
& \hspace{2.3cm} \times \left[ {}_2F_1(-\alpha,-\alpha;1;\rho_{\ell_1\ell_2}^2) -1 \right]
%= & \frac{1}{\alpha^2 N^2} \sum_{\ell_1=0}^{N-1} \sum_{\ell_2=0}^{N-1} e^{ik\frac{2\pi}{N}(\ell_1-\ell_2)}\lambda_{\ell_1}^\alpha\lambda_{\ell_2}^\alpha [\Gamma(\alpha+1)]^2 \nn\\ 
%& \hspace{2.3cm} \times \left\{ \sum_{n=0}^\infty \left[ \frac{(-\alpha)_n}{n!}\right]^2 \rho_{\ell_1\ell_2}^{2n} -1 \right\}\nn
\end{align}
where in the last equality we exploited \eqref{expect_hatPhi_power} and \eqref{cross_exp}.  
The computation can be continued as follows:
\begin{subequations}\label{seq_inequal}
\begin{align}
 & \var(\hat m_k) = |\var(\hat m_k)| \nn \\
\leq & \frac{[\Gamma(\alpha+1)]^2}{\alpha^2 N^2} \sum_{\ell_1=0}^{N-1} \sum_{\ell_2=0}^{N-1} \lambda_{\ell_1}^\alpha \lambda_{\ell_2}^\alpha  \left[ {}_2F_1(-\alpha,-\alpha;1;\rho_{\ell_1\ell_2}^2) -1 \right] \label{inequal_01} \\
\leq & \frac{[\Gamma(\alpha+1)]^2}{\alpha^2 N^2} \sum_{\ell_1=0}^{N-1} \sum_{\ell_2=0}^{N-1} \lambda_{\ell_1}^\alpha \lambda_{\ell_2}^\alpha  \left[ {}_2F_1(-\alpha,-\alpha;1;1) -1 \right] \label{inequal_02} \\
= & \frac{C_1}{\alpha^2 N^2} \sum_{\ell_1=0}^{N-1} \lambda_{\ell_1}^\alpha \sum_{\ell_2=0}^{N-1} \lambda_{\ell_2}^\alpha \to \frac{C_1}{\alpha^2} \left( \int_{\Tbb} \Phi(\theta)^\alpha \frac{\d\theta}{2\pi} \right)^2 = \mathrm{const.} \label{equal_03}
\end{align}
\end{subequations}
as $N\to\infty$ where,
\begin{itemize}
	\item \eqref{inequal_01} holds because using the series expansion \eqref{hypergeom}, we see that the quantity
	\begin{equation*}
	{}_2F_1(-\alpha,-\alpha;1;\rho_{\ell_1\ell_2}^2) -1 = \sum_{n=1}^\infty \left[ \frac{(-\alpha)_n}{n!}\right]^2 \rho_{\ell_1\ell_2}^{2n}
	\end{equation*}
	is positive;
	\item \eqref{inequal_02} follows from the relation $0\leq \rho_{\ell_1\ell_2}^2 \leq 1$ for the correlation coefficient;
	\item \eqref{equal_03} results from an application of Gauss's summation theorem \cite[Eq.~15.1.20]{abramowitz1972handbook}: %\bin{[Bin: I changed the reference to a more standard one for special functions.]} % \cite{koepf1998hypergeometric},   
	%and thus we have:
	\begin{equation}\label{Gaussthm}
	%\begin{aligned}
	_2F_1(-\alpha,-\alpha; 1; 1) = \frac{\Gamma(2\alpha+1)}{[\Gamma(\alpha+1)]^2}
	%\left\{\begin{array}{ll}\frac{\Gamma(0.5)\Gamma(2\alpha+0.5)}{\Gamma(\alpha+0.5)^2},  & \ell_1,\ell_2\in\mathcal S \\ \frac{\Gamma(2\alpha+1)}{\Gamma(\alpha+1)^2}, & \hbox{otherwise.}\end{array}\right.\nn
	%\end{aligned}
	\end{equation}
	which holds since we have restricted ourselves to the case $0<\alpha<1$,
	and $C_1$ is the same constant as the one in \eqref{calc_var_step5}.
\end{itemize}

%At this point, we notice that, 

As revealed by \eqref{equal_03}, the upper bound for $\hat{m}_k$ is a constant that does not improve as $N$ increases. 
Although such a bound may not be tight, consistency of the generalized cepstral estimator \eqref{est_cep_gen} cannot be directly derived in this case.
%without additional assumptions. 
There are two ways to resolve this issue. 
One way is to appeal to a \emph{windowed periodogram} $\check{\Phi}$ instead of the unwindowed version $\hat{\Phi}$, as hinted at the end of \cite[Sec.~IV]{ephraim2005second}.
More precisely, one first introduces a suitable window function $w: \Zbb\to \Rbb$, see \cite[Subsec.~6.2.3]{priestley1981spectral}, and then modifies the correlogram \eqref{Phi_hat_correlo} as
\begin{equation}
\check{\Phi}(\theta) = \sum_{k=-N+1}^{N-1} w(k) \hat{c}_k e^{-ik\theta}.
\end{equation}
A modified estimator $\check{m}_k$ results if one replaces $\hat{\Phi}$ in \eqref{est_cep_gen} with the above $\check{\Phi}$.
Such a practice is intuitively reasonable since under quite general conditions, $\check{\Phi}$ is a consistent estimator of the true spectrum $\Phi$ \cite[Subsec.~6.2.4]{priestley1981spectral}, and one may expect that the conclusion of Theorem~\ref{thm_gen_cep_est} holds for $\check{m}_k$. A rigorous proof seems hard (if possible) as for example in the computation of \eqref{expect_hatPhi_power}, now $\check{\Phi}_\ell$ is a linear combination of $\hat{\Phi}_\ell$ and has a \emph{generalized chi-squared} distribution \cite{davies1980algorithm}. Unfortunately, the latter distribution in general does not admit a closed-form expression for the probability density function.

%\bin{%[REFORMULATE THIS PART]
	
%We shall however, not make more restrictive 
The other way is to make an additional assumption about the underlying random process $y_t$ such as the following one.

\begin{assumption}\label{assump_decay_correl}
	The variances of the normalized spectral components $\{\bar{Y}_\ell\}_{\ell=0}^{N-1}$ are asymptotically bounded, that is, the quantity
	\begin{equation}\label{max_lambda_l}
	\gamma_N:=\max_{0\leq \ell\leq N-1}\lambda_{\ell}
	\end{equation}
	satisfies $\lim_{N\rightarrow\infty }\gamma_N<\infty$.
	Moreover, the growth of the correlation coefficients is under control in the sense that
	\begin{equation}\label{cond1}
	\sum_{\ell_2=0}^{N-1}\rho_{\ell_1\ell_2}^{2}\leq f(N), \quad\forall \ell_1 \in \{0, 1, \dots, N-1\}
	\end{equation}	
	where $f(N)$ is a function such that
	\begin{equation}\label{growth_f_N}
	\lim_{N\rightarrow \infty} f(N)/N =0.
	\end{equation}
\end{assumption}

{In plain words, condition (\ref{cond1}) requires that the correlation between the spectral components is mild.} We then have the next result.

\begin{proposition}
	Under Assumptions~\ref{assump_independ} $($without the independence between spectral components$)$ and \ref{assump_decay_correl}, the generalized cepstral estimator \eqref{est_cep_gen} is mean-square consistent up to a constant multiplicative factor, that is, \eqref{ms_converg_hat_mk}-\eqref{ms_converg_hat_m0} hold.
\end{proposition}

\begin{proof}
Apparently, the calculations in \eqref{var_hat_m} and \eqref{seq_inequal} are all valid. We shall modify the inequalities after \eqref{inequal_01} and show that the power of Assumption~\ref{assump_decay_correl} can lead to the result $\var(\hat{m}_k)\to 0$ as $N\to \infty$.
	
Given the fact $0\leq \rho_{\ell_1\ell_2}^2\leq1$ and the condition \eqref{cond1}, we have that
$$\sum_{\ell_2=0}^{N-1}\rho_{\ell_1\ell_2}^{2n}\leq f(N),\quad \forall\ell_1\in \{0, \dots, N-1\} \ \text{ and }\ \forall n\geq 1.$$
Next we expand \eqref{inequal_01} using the series \eqref{hypergeom}:
%Since $C_{\ell_1\ell_2}-\Gamma(\alpha+1)^2\leq 0$ under condition (\ref{cond_alpha}), we have that  
\begin{subequations}
\begin{align}
\var(\hat m_k)
& \leq \frac{[\Gamma(\alpha+1)]^2}{\alpha^2 N^2} \sum_{\ell_1=0}^{N-1} \sum_{\ell_2=0}^{N-1} \lambda_{\ell_1}^\alpha \lambda_{\ell_2}^\alpha  \sum_{n=1}^\infty \left[ \frac{(-\alpha)_n}{n!}\right]^2 \rho_{\ell_1\ell_2}^{2n} \nn \\
%\var (&\hat m_k) \leq\frac{1}{\alpha^2 N^2} \sum_{\ell_1=0}^{N-1} \sum_{\ell_2=0}^{N-1}  \lambda_{\ell_1}^\alpha\lambda_{\ell_2}^\alpha \times \nn\\ &\hspace{0.3cm}\times \left|C_{\ell_1\ell_2}\sum_{n=1}^\infty \frac{(-\alpha)_n^2}{(\beta_{\ell_1\ell_2})_n}\frac{1}{n!}\rho_{\ell_1\ell_2}^{2n}+C_{\ell_1\ell_2}-\Gamma(\alpha+1)^2\right| \nn\\
& \leq\frac{[\Gamma(\alpha+1)]^2 \gamma_N^{2\alpha}}{\alpha^2 N^2} \sum_{n=1}^\infty \left[ \frac{(-\alpha)_n}{n!}\right]^2 \, \sum_{\ell_1=0}^{N-1} \sum_{\ell_2=0}^{N-1} \rho_{\ell_1\ell_2}^{2n}  \label{inequal_03} \\
& \leq \frac{[\Gamma(\alpha+1)]^2 M^{2\alpha}}{\alpha^2} \left[ _2F_1(-\alpha,-\alpha;1;1) -1 \right] \frac{f(N)}{N} \label{inequal_04} \\
& = \mathrm{const.} \times \frac{f(N)}{N} \to 0 \ \text{ by }\ \eqref{growth_f_N}, \nn
\end{align} 
\end{subequations}
where,
\begin{itemize}
	\item \eqref{inequal_03} follows from \eqref{max_lambda_l} and interchanging the order of summation,
	\item and in \eqref{inequal_04} the constant $M$ is an upper bound for $\gamma_N$ (which exists by Assumption~\ref{assump_decay_correl}) and we have utilized \eqref{cond1}.
\end{itemize}
%Moreover,
%\begin{align}
%\var (&\hat m_k)\leq
%\frac{C \gamma_N^\alpha}{\alpha^2 N^2}\sum_{\ell_1=0}^{N-1}  \lambda_{\ell_1}^\alpha  \sum_{\ell_2=0}^{N-1} \sum_{n=1}^\infty \frac{(-\alpha)_n^2}{(\beta_{\ell_1\ell_2})_n}\frac{1}{n!}\rho_{\ell_1\ell_2}^{2n}\nn\\
%&<\frac{C\gamma_N^\alpha}{\alpha^2 N^2} \sum_{\ell_1=0}^{N-1}  \lambda_{\ell_1}^\alpha  \sum_{\ell_2=0}^{N-1} \sum_{n=1}^\infty \frac{(-\alpha)_n^2}{n!}\left(\frac{1}{(0.5)_n}+\frac{1}{(1)_n}\right)\rho_{\ell_1\ell_2}^{2n}\nn\\
%&=\frac{C \gamma_N^\alpha}{\alpha^2 N^2} \sum_{\ell_1=0}^{N-1}  \lambda_{\ell_1}^\alpha  \sum_{n=1}^\infty \frac{(-\alpha)_n^2}{n!}\left(\frac{1}{(0.5)_n}+\frac{1}{(1)_n}\right)\sum_{\ell_2=0}^{N-1} \rho_{\ell_1\ell_2}^{2n}\nn\\
%&\leq\frac{C f(N)}{\alpha^2 N^2}\gamma_N^\alpha \sum_{\ell_1=0}^{N-1}  \lambda_{\ell_1}^\alpha   \sum_{n=1}^\infty \frac{(-\alpha)_n^2}{n!} \left(\frac{1}{(0.5)_n}+\frac{1}{(1)_n}\right)\nn\\
%&\leq\frac{C  f(N)}{\alpha^2 N^2}\gamma_N^\alpha \sum_{\ell_1=0}^{N-1}  \lambda_{\ell_1}^\alpha  \times\nn\\ &\hspace{0.3cm}\times[{\,}_2F_1(-\alpha,-\alpha;0.5;1)+{\,}_2F_1(-\alpha,-\alpha;1;1)-2].  \nn
%\end{align}
%
%
%Let $$\bar C:=\frac{\Gamma(0.5)\Gamma(2\alpha+0.5)}{\Gamma(\alpha+0.5)^2}+ \frac{\Gamma(2\alpha+1)}{\Gamma(\alpha+1)^2} -2.$$
%Therefore, we obtain
%\begin{align}
%\var (&\hat m_k)\leq \frac{C \bar C}{\alpha^2 }  \frac{f(N)}{  N} \gamma_N^\alpha  \frac{1}{N} \sum_{\ell_1=0}^{N-1}  \lambda_{\ell_1}^\alpha \rightarrow 0 \hbox{ as } N\rightarrow \infty\nn.
%\end{align}
\end{proof}

%}

%{\mz
	
%{\em Example:} 

\begin{example}	
Although the condition in (\ref{cond1}) is not easy to check in practice, %on the other hand 
it seems to be a reasonable requirement for stationary stochastic processes generated by a {stable rational shaping} filter driven by white noise. Indeed, consider the ARMA (autoregressive moving-average) process $y_t=W(z)e_t$ where $e_t$ is {a} normalized white Gaussian noise, the shaping filter is defined as  
$$W(z)=\frac{z^2-z+0.8}{z^2-1.6z+0.81}=\sum_{t=0}^\infty w_t z^{-t},$$ 
and $w_t$ denotes the impulse response of the shaping filter. Then, it is not difficult to see that 
\begin{align*} \E &[\bar Y_{\ell_1}\bar Y_{\ell_2}^*]\nn\\ & =\frac{1}{N}\sum_{t_1, t_2=0}^{N-1} \sum_{s=0}^{\min\{t_1,t_2\}}e^{-i\frac{2\pi}{N} (\ell_1 t_1-\ell_2 t_2)} (w_{t_1-s}w_{t_2-s}).\end{align*}
Based on the computation with this formula, Figure \ref{fig_condtion} shows the value of {the sum $\sum_{\ell_2=0}^{N-1}\rho_{\ell_1\ell_2}^{2}$}
%\begin{align}\label{check_rate}\sum_{\ell_2=0}^{N-1}\frac{|\E[\bar Y_{\ell_1}\bar Y_{\ell_2}^*]|^2}{\E[|\bar Y_{\ell_1}|^2] \E[|\bar Y_{\ell_2}|^2]}\end{align}
as a function of $N$ with $\ell_1=3N/4$ (i.e., it corresponds to the angular frequency $\theta_1=3\pi/2$). A similar behavior has been found for different values of $\ell_1$. This experiment suggests that the sum {of the modulus squared correlation coefficients} does not grow as $N$ approaches infinity and thus conditions \eqref{cond1}-\eqref{growth_f_N} are satisfied.
%[MATTIA: Perhaps we could also observe that $Y_\ell$'s are weakly correlated as $N\rightarrow\infty$, so one would expect that only element in (\ref{cond1}) is different to zero, the remaining ones are infitesimal. Clearly, the infinite sum of infinitesimal terms does not necessarily converge or diverge slower than $N$. In this case, however, it seems to happen. ]

\begin{figure}
	\includegraphics[width=0.5\textwidth]{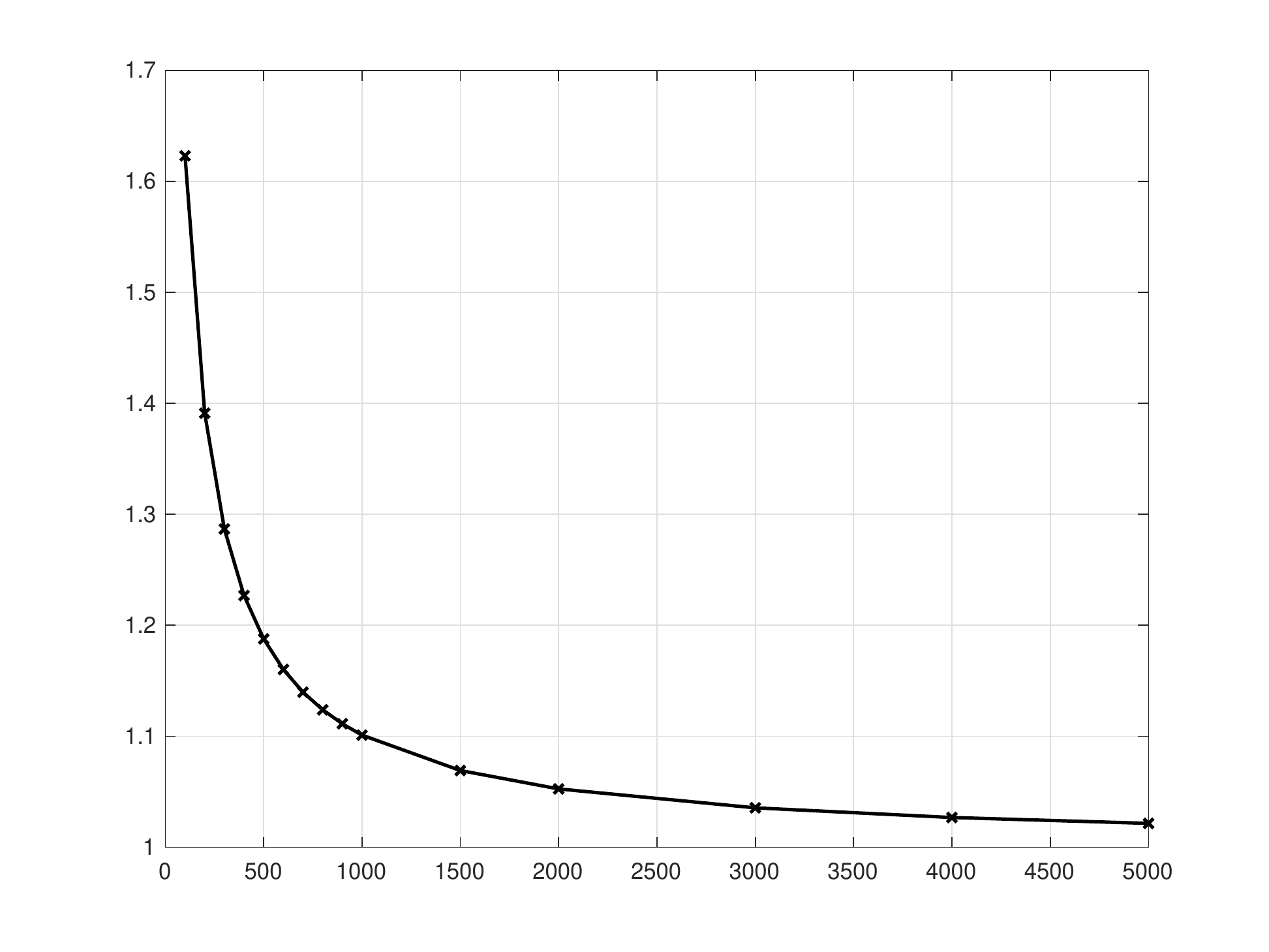}
	\caption{The partial sum of the correlation coefficients in \eqref{cond1} with $\ell_1=3N/4$ as a function of $N$. The crosses {represent} the values of the partial sum that have been computed.}\label{fig_condtion}
\end{figure}

\end{example} 
 
%\bin{Bin: The condition \eqref{cond1} for the correlation coefficients seems quite strong because a consequence is that asymptotically, $\rho(\bar \theta,\theta)=0$ almost everywhere. Moreover, as you noticed, it seems hard to justify such an assumption even in a simple AR$(1)$ model. Therefore, I would avoid making such an assumption and rather turn to \emph{windowed} periodograms which appear to be a more natural and reasonable choice given the fact the corresponding theory has been well established at least in the 1-d case \cite{priestley1981spectral}.}

%}

\section{Extension to Multidimensional Random Fields}\label{sec:multidim}

We want to point out that the results obtained in Secs.~\ref{sec:independ} and \ref{sec:depend} transit smoothly to the multidimensional case. In order to see this, let $y_{\tb}$ be a zero-mean stationary $d$-dimensional complex-valued random field where $\tb=(t_1,\dots,t_d)\in\Zbb^d$. Given a vector $\Nb=(N_1,\dots,N_d)\in\Nbb_+^d$, define a finite index set
\begin{equation}
\Zbb^d_{\Nb} := \left\{ (t_1,\dots,t_d) \,:\, 0\leq t_j\leq N_j-1,\ j=1,\dots,d \right\},
\end{equation}
which is just the integer grid of a $d$-dimensional box with a cardinality $|\Nb|:=\prod_{j=1}^{d}N_j$.
Suppose that we have obtained samples of the field $\{y_{\tb} : \tb\in\Zbb^d_{\Nb}\}$. The spectral domain now becomes $\Tbb^d$ containing frequency vectors $\thetab=(\theta_1,\dots,\theta_d)$.
For the computation of the DFT, the multidimensional counterpart of \eqref{T_grid} for the discretization of $\Tbb^d$ is given as
\begin{equation}\label{Td_grid}
\Tbb^d_\Nb:=\left\{ \left( \frac{2\pi}{N_1}\ell_1,\dots,\frac{2\pi}{N_d}\ell_d \right) \in \Tbb^d : \lb \in\Zbb_\Nb^d \right\}.
\end{equation}
Obviously, the map $\lb\mapsto\thetab_{\lb} := ({2\pi\ell_1}/{N_1},\dots,{2\pi\ell_d}/{N_d}) $ is a bijection from $\Zbb_\Nb^d$ to $\Tbb_\Nb^d$.
The spectral components are given by the multidimensional DFT
\begin{equation}
Y_{\lb} := Y(\thetab_{\lb}) = \sum_{t\in\Zbb_\Nb^d} y_{\tb} \, e^{-i\innerprod{\tb}{\thetab_{\lb}}},\quad \thetab_{\lb}\in\Tbb_\Nb^d
\end{equation}
whose normalized version is $\bar{Y}_{\lb} := \frac{1}{\sqrt{|\Nb|}}Y_{\lb}$. The periodogram is still given by $\hat\Phi_{\lb}=|\bar{Y}_{\lb}|^2$.
%REMOVE the $\zeta$ notation and simply write instead $e^{i\innerprod{\kb}{\thetab}}$ with $\thetab\in\Tbb_\Nb^d$...
%Let $\Nb$ denote the vector $(N_1,N_2,\dots,N_d)\in\Nbb_+^d$, where the component $N_j$ stands for the number of equal-length partitions of the interval $[0, 2\pi)$ in the $j$-th dimension of $\Tbb^d$. 
%We can now discretize the domain $\Tbb^d$ as
%so the number of grid points is also $|\Nb|$.
%Moreover, let us introduce the symbol $\zetab_{\lb}:=\left(\zeta_{\ell_1},\dots,\zeta_{\ell_d}\right)$ for a point in the discrete $d$-torus with $\zeta_{\ell_j}=e^{i2\pi\ell_j/N_j}$ and define $\zetab^{\kb}_{\lb}:=\prod_{j=1}^{d}\zeta_{\ell_j}^{k_j}$ which is just the complex exponential function $e^{i\innerprod{\kb}{\thetab}}$ evaluated at a grid point in $\Tbb_\Nb^d$.
The multidimensional version of Theorem~\ref{thm_gen_cep_est} is summarized below.

\begin{proposition}\label{prop_gen_cep_est_multidim}
	Assume that the normalized spectral components $\{\bar Y_{\lb} : \lb\in\Zbb_\Nb^d\}$ of the random field $\{y_{\tb} : \tb\in\Zbb_\Nb^d\}$ are independent zero-mean complex (circular) Gaussian random variables such that $\var \bar{Y}_{\lb} = \E\, \hat{\Phi}_{\lb} = \lambda_{\lb}>0$. Assume in addition that the underlying true spectral density $\Phi(\thetab)$ of the field $y_{\tb}$ satisfies some regularity properties as discussed in Remarks~\ref{rem_regular_Phi} and \ref{rem_Lp_Phi}.
    Then, the estimator 
    \begin{equation}\label{est_cep_gen_multidim}
    \hat{m}_{\kb} = \begin{cases}
    \frac{1}{\alpha |\Nb|} \sum_{\lb\in\Zbb_\Nb^d} e^{i \innerprod{\kb}{\thetab_{\lb}}} \hat{\Phi}_{\lb}^\alpha & \text{if } \kb\neq \zerob, \\
    \\
    \frac{1}{\alpha} \left( \frac{1}{|\Nb|} \sum_{\lb\in\Zbb_\Nb^d} \hat{\Phi}_{\lb}^\alpha -1 \right) & \text{if } \kb=\zerob
    \end{cases}
    \end{equation}
    of the generalized cepstral coefficients 
    \begin{equation}\label{def_cep_gen_multidim}
    m_{\kb} = \begin{cases}
    \frac{1}{\alpha} \int_{\Tbb^d} e^{i \innerprod{\kb}{\thetab}} \Phi(\thetab)^\alpha \d\mu(\thetab) & \text{if } \kb\neq \zerob, \\
    \\
    \frac{1}{\alpha} \left[ \int_{\Tbb^d} \Phi(\thetab)^\alpha \d\mu(\thetab) -1 \right] & \text{if } \kb=\zerob
    \end{cases}
    \end{equation}
    is consistent up to a constant multiplicative factor.
    Here $\d\mu(\thetab)=\frac{1}{(2\pi)^d}\prod_{j=1}^{d} \d\theta_j$
    %\end{equation}
    is the normalized Lebesgue measure on $\Tbb^d$.
    More precisely, we have for $\kb\neq \zerob$ the convergence
	\begin{equation}\label{ms_converg_hat_mk_multidim}
	C \hat{m}_{\kb} \xrightarrow{\text{m.s.}} m_{\kb}, \ \text{ and } \ 
	C\hat{m}_{\zerob} + {(C-1)}/{\alpha} \xrightarrow{\text{m.s.}} m_{\zerob}
	\end{equation}
%	and
%	\begin{equation}\label{ms_converg_hat_m0_multidim}
%	$C\hat{m}_{\zerob} + {(C-1)}/{\alpha} \xrightarrow{\text{m.s.}} m_{\zerob}$
%	\end{equation}
	as $\min(\Nb)\to\infty$, where the constant $C = 1/\Gamma(\alpha+1)$.
\end{proposition}

%{\mz [MATTIA: I would add some hint about the differences (perhaps a couple of sentences) about the proof between the unidimensional case and the multidimensional one. For instance, one would wonder the multidimensional counterpart of (\ref{converg_var_hat_mk}). \bin{Bin: Will do!}]}

%\bin{
	
Notice that the proof of the above proposition is almost identical to that in the unidimensional case (hence omitted), as the computation of the quantities such as \eqref{expect_hatPhi_power} and \eqref{var_hatPhi_l} is left unchanged.
The only difference in the multidimensional case is that we need to use multiple indices for the summation and write multivariable integrals. For example, \eqref{converg_var_hat_mk} should be rewritten as
\begin{equation}
\frac{1}{|\Nb|} \sum_{\lb\in\Zbb_\Nb^d} \lambda_{\lb}^{2\alpha} \to \int_{\Tbb^d} \Phi(\thetab)^{2\alpha} \d\mu(\thetab).
\end{equation}

As for the case with correlated spectral components, we can either use windowing techniques as explained in \cite[Sec.~9.7]{priestley1981spectral} for the 2-d case and in \cite{engels2017advances,ZFKZ2019fusion} for 3-d radar signal processing, or make an argument along the lines of Section~\ref{sec:depend} where conditions in Assumption~\ref{assump_decay_correl} are now replaced by  
\begin{align}
& \lim_{\min(\Nb)\to \infty} \gamma_{\Nb} <\infty, \quad
\sum_{\lb_2\in\Zbb_\Nb^d}\rho_{\lb_1\lb_2}^{2}\leq f(\Nb), \quad\forall \lb_1\in\Zbb_\Nb^d,\nn
\end{align}
where 
$$  \gamma_{\Nb}:=\max_{\lb\in \Zbb_\Nb^d } \,\lambda_{\lb}, \quad
\lim_{\min(\Nb)\to \infty} f(\Nb)/|\Nb|=0.$$

%}

\section{Cascade System Identification}\label{sec:sysid}

In this section, we show that our generalized cepstral estimator can be integrated into a spectral estimation procedure for the identification of cascade (possibly multidimensional) linear stochastic systems with \emph{identical} subsystems, see Fig.~\ref{fig:cascade_linear_system} where $W(z)$ denotes the transfer function of each subsystem, $y_t$ is the output, and $e_t$ is the white noise input.
%A laboratory example is the double-tank system where the output flow of the first fluid tank is fed into the second tank. {\mz [This example is not so much credible. Indeed, I do not believe that there exists a cascade of tanks driven by white noise in a real example. This problem can be partially solved as follows:
In practice, of course the system need not be driven by a white noise but rather has an input $u_t$ which can be modeled as a stationary stochastic process. The output of the cascade system then becomes (symbolically) $y_t=[W(z)]^2 u_t$. Assume that we are able to measure both $y_t$ and $u_t$.
At that point one can estimate the power spectral density of the input and thus also its shaping filter $\hat W_u$. After that, we can filter the output to obtain $\tilde y_t=[\hat W_u(z)]^{-1} y_t$ which approximately admits the model of Fig.~\ref{fig:cascade_linear_system}. In this sense we have \emph{whitened} the input process. Therefore, it is not restrictive to consider a white noise input in the cascade system above.
%in Fig.1 with the explanation above. Moreover, I would avoid to start with the general cascade (that we don't use after and so one would wonder why it is important the general cascade in this reasoning), but rather the cascade of identical systems.
%]} 

A cascade of identical systems is {commonly used} to model reactors in chemical engineering: each subsystem corresponds to a continuous stirred-tank reactor (CSTR) and the aim of the cascade is to increase the residence time distribution (RTD) of the reactor \cite{renken2014microstructured}. Moreover, the multidimensional setting allows to model the fact that influent and effluent concentration distributions are non-uniform, see for instance \cite{vavilin2007anaerobic}.

%The underlying assumption of our approach is that 
%all the subsystems in the cascade, such as $W_1$ and $W_2$ in the figure, are \emph{identical}. Moreover, 
{We assume that} the number of subsystems which we call $\nu$, is \emph{a priori known} so that the parameter $\alpha$ is specialized as $1-\frac{1}{\nu}$ (see also Example~\ref{ex_alpha_nu} in Section~\ref{sec:independ}). The latter choice of the parameter turns out to correctly recover the number of subsystems. At the heart of our system identification approach lies a spectral estimation problem presented in \cite{Zhu-Zorzi-2021-cepstral,Zhu-Zorzi-21} which we shall briefly recall next.

\begin{figure}[!t]
	%\begin{center}
	\centering
	\tikzstyle{int}=[draw, minimum size=2em]
	\tikzstyle{init} = [pin edge={to-,thin,black}]
	\begin{tikzpicture}[node distance=2cm,auto,>=latex']
	\node [int] (a) {$W(z)$};
	\node (b) [left of=a, coordinate] {};
	\node [int] (c) [right of=a] {$W(z)$};
	\node (d) [right of=c] {};
	%\node [coordinate] (end) [right of=c, node distance=2cm]{};
	\path[->] (b) edge node {$e_{t}$} (a);
	\path[->] (c) edge node {$y_{t}$} (d);
	\draw[->] (a) edge node {} (c) ;
	\end{tikzpicture}
	%\end{center}
	\caption{A cascade linear stochastic system with two identical subsystems.}
	\label{fig:cascade_linear_system}
\end{figure}
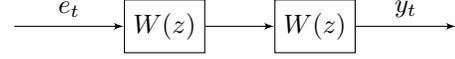

%DESCRIBE the problem formulation in \cite{Zhu-Zorzi-2021-cepstral,Zhu-Zorzi-21} and use the well-posedness result to show the convergence of the estimator to the true model...

Suppose that from some underlying random field we are given a number of covariances $\{c_{\kb} : \kb\in\Lambda\}$ and generalized cepstral coefficients $\{m_{\kb} : \kb\in\Lambda_0\}$, and from these numbers we want to reconstruct the spectrum $\Phi(\thetab)$, a nonnegative function in $L^1(\Tbb^d)$. Here the index set $\Lambda$ can be any finite set that contains the zero vector and is symmetric with respect to the origin, namely $\kb\in \Lambda \implies -\kb\in\Lambda$. The other index set $\Lambda_0 := \Lambda\backslash\{\zerob\}$ so $m_{\zerob}$ is excluded for technical reasons. Inspired by the famous \emph{Maximum Entropy} method \cite{burg1975maximum}, we set up the following general optimization problem:
\begin{subequations}\label{primal_pb}
	\begin{alignat}{2}
	& \underset{\Phi\geq 0}{\max}
	& \quad & \Hbb_\nu(\Phi) := \frac{\nu^2}{\nu-1} \left(\int_{\Tbb^d}\Phi(\thetab)^{\frac{\nu-1}{\nu}} \d\mu(\thetab) -1 \right) \label{nu_entropy} \\
	&\hspace{0.2cm} \text{s.t.}
	& \quad & c_\kb=\int_{\Tbb^d}e^{i\innerprod{\kb}{\thetab}}\Phi(\thetab) \d\mu(\thetab) \;\ \forall \kb\in\Lambda, \label{constraint_cov} \\
	& 
	& \quad & m_\kb= \frac{\nu}{\nu-1}\int_{\Tbb^d}e^{i\innerprod{\kb}{\thetab}}\Phi(\thetab)^{\frac{\nu-1}{\nu}} \d\mu(\thetab) \;\ \forall \kb\in\Lambda_0, \label{constraint_cep}
	\end{alignat}
\end{subequations}
where the objective functional $\Hbb_\nu(\Phi)$ is a generalized version of the entropy of $\Phi$ derived from the $\alpha$-divergence \cite{Z-14rat,zorzi2015interpretation}. The above formulation can be viewed as a considerable generalization of the simultaneous covariance-cepstral extension theory in \cite{enqvist2004aconvex,RKL-16multidimensional} as their problem can be recovered from ours by letting $\nu\to 1$. The optimization problem in \eqref{primal_pb} is infinite-dimensional, and it is usually more convenient to work with its finite-dimensional dual problem. Following the analysis in \cite{Zhu-Zorzi-2021-cepstral}, a suitable form of \emph{regularization} is needed in order to promote a unique positive solution to the dual problem, so that the resulting primal solution $\Phi$ is a strictly positive spectral density which is desired in many practical applications. More specifically, the regularized dual optimization problem is formulated as:
\begin{align}\label{reg_dual_pb}
& \underset{\pb,\, \qb}{\min\,}\ J_{\nu,\lambda}(\pb,\qb) := J_{\nu}(\pb,\qb)+\frac{\lambda}{\nu-1} \int_{\Tbb^d}\frac{1}{P(\thetab)^{\nu-1}} \d\mu(\thetab)  \nn\\
& \hbox{ s.t. } P(\thetab) \geq 0 , \quad Q(\thetab)\geq 0, \quad \forall \thetab\in\Tbb^d
%\in\overline{\Pfrak}_{+,o},\quad Q\in\overline{\Pfrak}_+,
\end{align}
%} 
where, the dual variables (the Lagrange multipliers) $\{q_{\kb} : \kb\in\Lambda\}$ and $\{p_{\kb} : \kb\in\Lambda_0\}$ correspond to the covariance and cepstral constraints \eqref{constraint_cov} and \eqref{constraint_cep}, respectively, 
$P(\thetab):= 1 + \sum_{\kb\in\Lambda_0} p_\kb e^{-i\innerprod{\kb}{\thetab}}$ and $Q(\thetab):=\sum_{\kb\in\Lambda} q_\kb e^{-i\innerprod{\kb}{\thetab}}$ are trigonometric polynomials, and $\lambda>0$ is the regularization parameter. The first term in the dual objective function $J_{\nu,\lambda}(\pb,\qb)$ is the unregularized dual function
\begin{equation}\label{dual_func}
J_{\nu}(\pb,\qb) := \frac{1}{\nu-1}\int_{\Tbb^d}  \frac{P(\thetab)^\nu}{Q(\thetab)^{\nu-1}} \d \mu(\thetab) +\langle \qb, \cb\rangle  -\langle \pb,\mb\rangle,
\end{equation}
where the real-valued inner product $\langle \qb, \cb\rangle := \sum_{\kb\in\Lambda} q_{\kb} c_{\kb}^*$.

Under a suitable feasibility assumption for the covariance (multi)sequence $\{c_{\kb} : \kb\in\Lambda\}$ (see \cite[Assumption~5.1]{Zhu-Zorzi-2021-cepstral}), if the integer parameter $\nu$, the number of subsystems in the cascade, satisfies the condition $\nu\geq \frac{d}{2}+1$, then the regularized dual problem \eqref{reg_dual_pb} admits a unique solution $(\hat{\pb},\hat{\qb})$ such that the corresponding polynomials $\hat P$ and $\hat Q$ are both strictly positive.
%\in \Pfrak_{+,o}\times \Pfrak_+$.
Moreover, the positive rational function 
\begin{equation}\label{Phi_optimal}
\hat{\Phi}_\nu(\thetab) = [\hat{P}(\thetab)/\hat{Q}(\thetab)]^\nu
\end{equation}
matches the given covariances $\{c_\kb : \kb \in\Lambda\}$ exactly and the generalized cepstral coefficients $\{m_\kb : \kb \in\Lambda_{0}\}$ approximately. Clearly we have set $p_{\zerob}=1$ in order to avoid trivial cancellations between $P$ and $Q$ in the rational function $\Phi_{\nu}$.

From the perspective of numerical computation, it is better to \emph{discretize} the problems \eqref{reg_dual_pb} on a grid similar to \eqref{Td_grid} but of a different size $\Kb=(K_1,\dots,K_d)\in\Nbb_+^d$. The formulas will be similar to the ones above and are omitted here. For details see \cite[Section 8]{Zhu-Zorzi-2021-cepstral} where it is also explained that the solution to the discrete problem is a reasonable approximation of the solution to the corresponding continuous problem when the grid size $\Kb$ is sufficiently large (componentwise).

To summarize, our system identification procedure can be divided into four steps:

\noindent\textbf{(i)} Feed the system with normalized white noise $e_{\tb}$ and collect the output random field $\{y_{\tb} : \tb\in\Zbb_\Nb^d\}$;\\
\textbf{(ii)} Choose an index set $\Lambda$ and estimate $\{c_{\kb} : \kb\in\Lambda\}$ via \eqref{cov_estimates} adapted to multidimensional random fields (see e.g., \cite{ZFKZ2019M2}) and $\{m_{\kb} : \kb\in\Lambda_0\}$ via \eqref{est_cep_gen_multidim} corrected with the constant $C$ in \eqref{ms_converg_hat_mk_multidim};\\ % (where $\nu$ has been fixed)
\textbf{(iii)} Fix a grid size $\Kb$ and solve the discrete version of the regularized dual optimization problem \eqref{reg_dual_pb} given {$c_{\kb}$'s and $m_{\kb}$'s} in Step (ii) and a regularization parameter {$\lambda>0$};\\
\textbf{(iv)} Factor the optimal spectrum $\hat\Phi(\thetab)$ in order to obtain a transfer function $\hat{W}(\zb)$ where $\zb:=(z_1,\dots,z_d)$.

We must point out that there is a significant technical difficulty in the last step of the procedure concerning \emph{spectral factorization}.
Since our optimal spectrum is a \emph{rational} function, we end up factoring positive Laurent trigonometric polynomials of several variables into \emph{one square}, which is in general an impossible task.\footnote{It has been shown that \emph{sum-of-squares} factorization is always possible for multivariate Laurent polynomials that are strictly positive on the multi-torus \cite{dritschel2004factorization}. However, the factors in general have degrees larger than the original Laurent polynomial.} Fortunately in the $2$-d case, \cite[Theorems 1.1.1 \& 1.1.3]{geronimo2004positive} have given a sufficient and necessary condition to check such factorability and an explicit formula to compute the factor when the factorization is possible. The nontrivial part of the condition states that a certain submatrix of the full covariance matrix should have a specific low rank. 
%Notice also that a ``factor'' $a(\thetab)$ can nonetheless be computed from a positive Laurent polynomial $P(\thetab)$ even when the aforementioned rank condition is not met. However, as discussed in \cite[Section~7]{RKL-16multidimensional}, such an operation results in a large difference between $P(\thetab)$ and $|a(\thetab)|^2$.

%\bin{[WRITE down a consistency Proposition for the identified system parameters?]}

In addition, we want to stress that the regularization employed in \eqref{reg_dual_pb} has the power of promoting strict convexity\footnote{The unregularized dual function \eqref{dual_func} is only convex but not strictly convex \cite{Zhu-Zorzi-2021-cepstral}.} and well-posedness of the dual problem \cite{Zhu-Zorzi-2021-cepstral}. Of course, the regularization term does not appear naturally in the system identification problem. Thus in order to get a reasonable consistency result for the identified model, we should take the regularization parameter $\lambda$ small as it directly controls the error in generalized cepstral matching. The next proposition expresses the idea.

\begin{proposition}\label{prop_identifiability}
	Assume that the model class is correctly specified, i.e., the unknown system transfer function $W(\zb)$ is such that true spectrum of the output random field $y_{\tb}$ $($when the input $e_{\tb}$ is a normalized white noise$)$ falls within the class of rational functions having the form \eqref{Phi_optimal}. Then as the number of samples $\min(\Nb)\to\infty$ and the regularization parameter $\lambda\to 0$, the optimal solution to \eqref{reg_dual_pb} recovers the true spectrum of $y_{\tb}$ as well as the correct system parameters in $W(\zb)$ after spectral factorization.	
\end{proposition}

\begin{proof}
	The proposition is a direct consequence of Theorem~6.1 in \cite{Zhu-Zorzi-2021-cepstral}.
\end{proof}

In plain words the above result says: given a $d$-dimensional system, which is a cascade of $\nu\geq d/2+1$ identical {subsystems}, it is possible to construct a consistent estimator of it {using some second-order statistics of the output process}.

\section{Simulation Results}\label{sec:sims}

In this section, we present simulation results in $1$-d and $2$-d cascade system identification because theoretical results and algorithms for polynomial spectral factorization are available only when the dimension is such that $d\leq 2$, as highlighted in the previous section.

\subsection{Identification of $1$-D dynamical systems}\label{subsec:1-d_sims}

Consider a $1$-d linear time-invariant (LTI) system described by the transfer function $W(z)$. Furthermore, let $W$ have the cascade structure that corresponds to our optimal spectrum \eqref{Phi_optimal}, namely $W(z) = \left[ W_1(z) \right]^\nu$ where each subsystem
\begin{equation}
W_1(z) = \frac{b(z)}{a(z)}= \frac{\sum_{k=0}^n b_k z^{-k}}{\sum_{k=0}^n a_k z^{-k}}.
\end{equation}
Here the index set is $\Lambda_+ = \{0, 1, \dots, n\}$ and we take $n=2$ which corresponds to a second-order system for simplicity.

Let us write $a(\theta)\equiv a(e^{i\theta})$ for the value of the polynomial $a(z)$ on the unit circle.
If the white noise input $e_{t}$ has unit variance, then the spectral density of the output process $y_{t}$ is $\Phi(\theta)=[P(\theta)/Q(\theta)]^\nu$ where
$P(\theta) = |b(\theta)|^2$, $Q(\theta) = |a(\theta)|^2$.
We shall take the integer $\nu=3$ which is known by our identification procedure.
%True system, a cascade of three identical copies of a second-order linear system. 
Moreover, we shall parametrize the polynomials in terms of their roots, that is, 
$a(z)=(1-r_{a,1} z^{-1})(1-r_{a, 2} z^{-1})$ and $b(z)=b_0(1-r_{b,1} z^{-1})(1-r_{b,2} z^{-1})$ with $b_0$ a normalization constant such that $P(z)=b(z)b(z^{-1})$ has $p_0=1$. The true system parameters are specified as
$r_{a,1}, r_{a,2} = 0.5 e^{\pm i\pi/3}$ (complex conjugate roots) and $r_{b,1}, r_{b,2} = -0.8$ and $0.6$ (real roots), so the system $W_1$ (also $W$) has \emph{real} parameters and is clearly stable and minimum-phase. This in turn gives
\begin{equation}\label{true_sys_paremters_1d}
\begin{aligned}
\ab & = \bmat a_0 & a_1 & a_2\emat = \bmat 1 & -0.5 & 0.25 \emat \ \text{ and } \ \\
\bb & = \bmat b_0 & b_1 & b_2\emat = \bmat 0.8872 & 0.1774 & -0.4259 \emat.
\end{aligned}
\end{equation}

Next, we follow the system identification procedure listed in the previous section, and demonstrate how to estimate the system parameters $\ab$ and $\bb$ from the samples of the output process $y_{t}$. 
Such samples of size $N$ are generated by feeding the normalized \emph{real} white noise $e_{t}$ to the true system (so $y_{t}$ is also real).
Then the covariances and generalized cepstral coefficients of the output $y_{t}$ with the indices in the set $\Lambda=\{ k\in\Zbb \,:\, -2\leq k\leq 2\}$ are estimated using \eqref{cov_estimates} and \eqref{est_cep_gen}.
Due to the symmetry, we only need to compute
\begin{equation*}
\hat{\cb} = \bmat\hat c_0 & \hat c_1 & \hat c_2\emat \ \text{ and } \ \hat{\mb} = \bmat \hat m_1 & \hat m_2 \emat.
\end{equation*}
In particular, the constant $C$ in Theorem~\ref{thm_gen_cep_est} for correcting the generalized cepstral estimation is equal to $1/\Gamma(\frac{5}{3})\approx 1.1077$ with $\alpha=\frac{\nu-1}{\nu} = 2/3$. The estimation error is defined as $\|[\hat{\cb}, \hat{\mb}]-[\cb, \mb]\|$ where the true $\cb$ and $\mb$ are evaluated using the true spectrum. Such errors against different sample sizes $N$ are shown in the orange line\footnote{The line is drawn to facilitate observation and has no meaning here.} of Fig.~\ref{subfig:err_sys_para_1d} in the double-logarithmic scale (base $10$). {At the same time, the yellow line indicates the errors when $\hat{\mb}$ is not corrected by the constant $C$}. To be more clear, one time series $y_t$ of length $N_{\max}=10^4$ is generated first and simulations are carried out using the same series truncated to different lengths $N$. One can see the general trend {of the orange line} that the error reduces as the sample size grows. The unique exception at $N=500$ can be explained as random fluctuation before the convergence of the estimators. In this particular example, the uncorrected generalized cepstral estimator begins to produce a large error (relative to that of the corrected version) after $N=2500$. Such an erratic behavior will be more apparent in the $2$-d example presented in the next subsection.
%\bin{[Bin: The estimation errors here drop much more slowly than those in the 2-d example. I would explain this by the fact that the 1-d system has a higher order. More precisely, in Subsection~\ref{subsec:2-d_sims} the 2-d subsystem $W_1(\zb)$ in the cascade  is of order one in each dimension and $\nu=2$ (a total order of 2 for $W(\zb)$). In contrast, here the 1-d subsystem $W_1(z)$ is of order 2 and $\nu=3$ (a total order of 6 for $W(z)$). However, I'm not sure whether to include these lines in the paper or not.]} {\mz I propose not to include this part.}

Then the regularized dual problem \eqref{reg_dual_pb} is solved on a discrete grid of size $K=N$ with the estimated $\hat{\cb}$, $\hat{\mb}$, and $\lambda=10^{-6}$ using a gradient descent algorithm.
%Clearly, the variables $\pb$ and $\qb$ can be arranged in the same way as that for $\hat{\mb}$ and $\hat{\cb}$ so we have $9$ real variables in total (a small-scale problem). 
The algorithm is initialized at $q_{0}=1$ and the rest variables equal to $0$ which corresponds to constant polynomials $P(\theta)=Q(\theta)\equiv1$. 
The iterations terminate when the norm of the gradient is less than $10^{-6}$. 
Given the optimal polynomials $\hat{P}$ and $\hat{Q}$, we proceed to compute their factors $\hat{a}$ and $\hat{b}$ using standard methods for ($1$-d) polynomial spectral factorization, e.g., the Bauer method \cite{sayed2001survey}.
After a comparison with \eqref{true_sys_paremters_1d}, we compute the error $\|[\hat{\ab}, \hat{\bb}]-[\ab, \bb]\|$ which is shown in the blue line of Fig.~\ref{subfig:err_sys_para_1d} for different values of $N$. %\bin{[TBC...]}
%we can conclude that the system parameters are recovered with very small errors, 
The blue line shares the same trend with the orange line as the sample size $N$ increases, which is an expected result since the true system $W(z)$ belongs to the model class dictated by the solution form \eqref{Phi_optimal} of our optimization problem, as explained in Proposition~\ref{prop_identifiability}. Furthermore, in Fig.~\ref{subfig:true_vs_recon_spec_1d} we plot the true spectrum $\Phi(\theta)$ on the interval $[0,\pi)$ against the estimated spectra $\hat{\Phi}(\theta)$. The convergence of the estimated spectra to the true one is clearly observed.

\begin{figure}
	%	\centering
	\begin{subfigure}[b]{.48\columnwidth}
		\includegraphics[width=\linewidth]{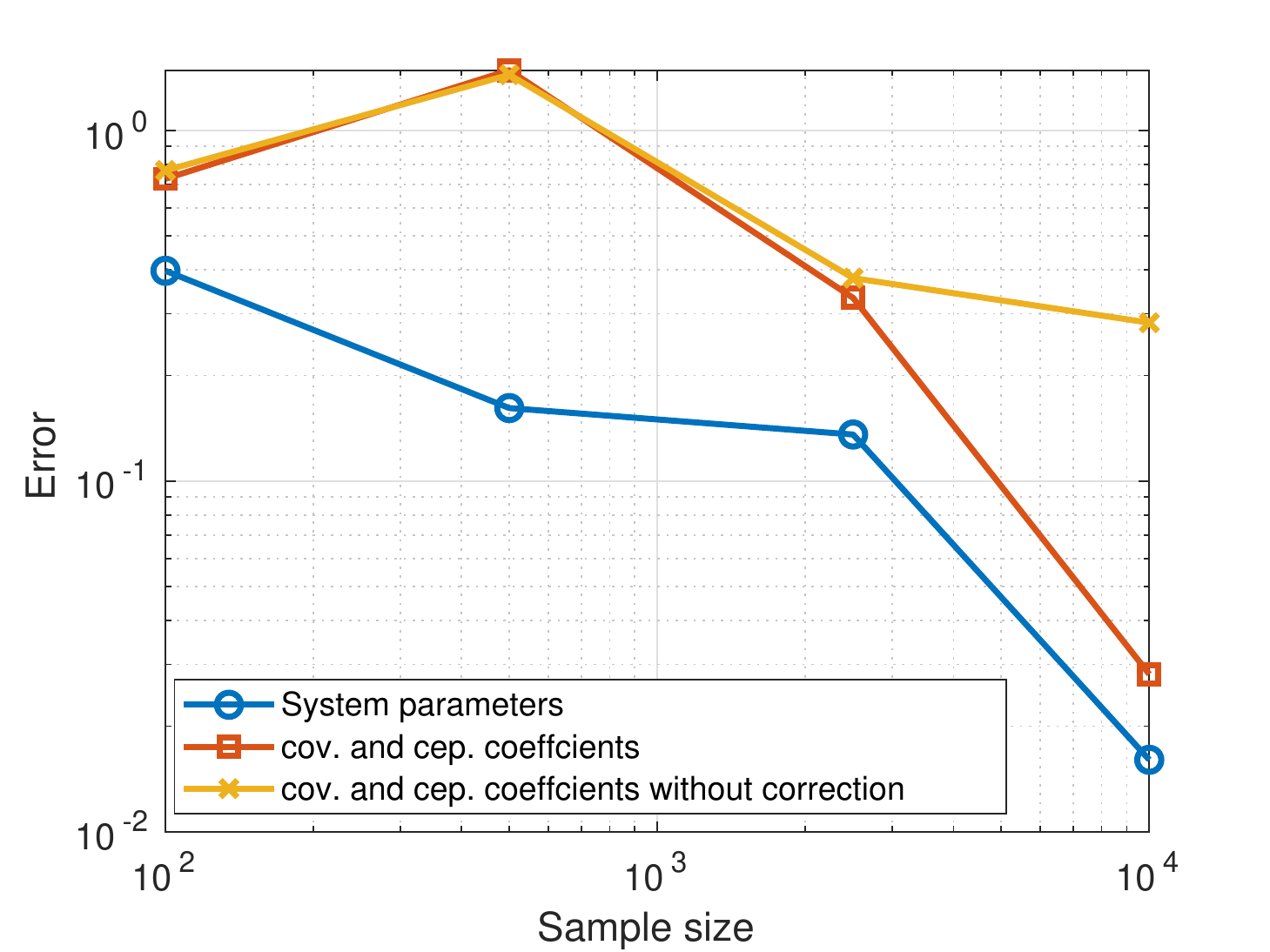}
		\caption{}
		\label{subfig:err_sys_para_1d}
	\end{subfigure}
	\hfill
	\begin{subfigure}[b]{.48\columnwidth}
		\includegraphics[width=\linewidth]{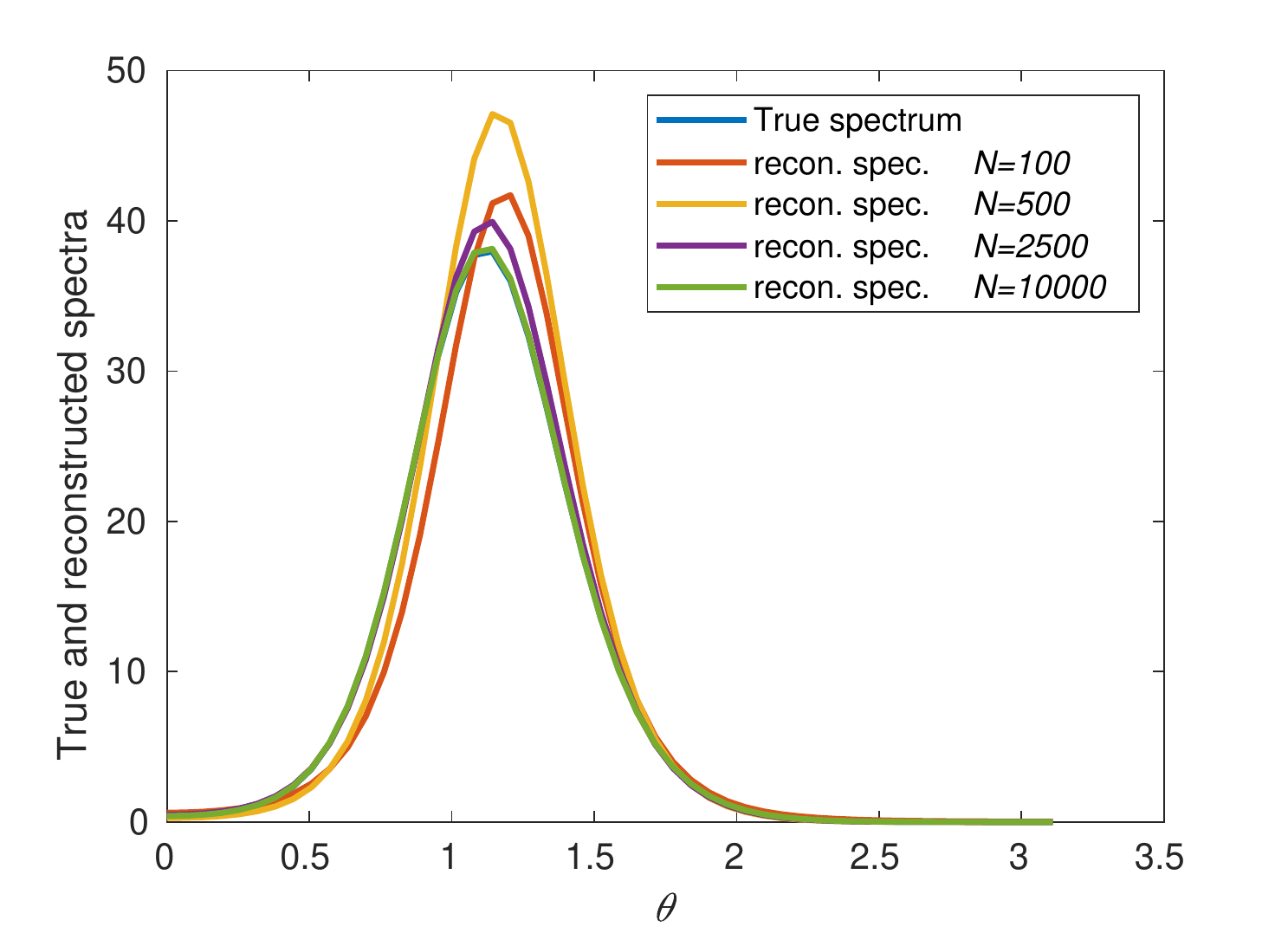}
		\caption{}
		\label{subfig:true_vs_recon_spec_1d}
	\end{subfigure}
	
	\caption{Simulation results in the $1$-d case: \emph{Left.} Estimation errors of the covariances and the generalized cepstral coefficients ({the orange line with correction and the yellow line without correction}) and the system parameters (blue line) when the sample size $N$ of the random process is equal to $100, 500, 2500, 10^4$. \emph{Right.} The true spectrum (blue line) $\Phi(\theta)$ of $y_{t}$ versus the estimated spectra $\hat{\Phi}(\theta)$ via solving the discrete version of the regularized dual optimization problem \eqref{reg_dual_pb} with $\lambda=10^{-6}$ and the sample sizes indicated above. Only the segment corresponding to $\theta\in[0, \pi)$ is shown due to the symmetry of the spectra of real processes. Notice that the estimated spectrum with $N=10^4$ virtually coincide with the true one.}
	\label{fig:sim_results_1d}
\end{figure}

\subsection{Modeling $2$-D random fields}\label{subsec:2-d_sims}

%{\mz [MATTIA: The example is fine with me. It is also good the motivation (that we should also write in the introduction). \bin{Yes!}
%	
%I would remove fig 2c 2d, because they do not say too much about consistency. 
%\bin{Note that consistency is an asymptotic result. I would keep Figs.~2(c, d) because they say that in order to reconstruct the spectrum to a fairly good precision, one only needs a data size equal to one hundred instead of e.g., one million.}
%
%It would be interesting to increase N (if possible). \bin{Given Figs.~2(c, d), certainly increasing $N$ will make things better, as also shown in Fig.~2(a).}
%
%Finally, should we also consider a unidimensional example? \bin{Bin: I think one example should be fine, just for the demonstration of the theory.}]}

In this subsection, we consider the problem of identifying a $2$-d LTI system $W(z_1,z_2)$ from samples of a planar random process.
The procedure will be similar to the previous subsection which we will describe briefly.
Let $W(\zb) = \left[ W_1(\zb) \right]^\nu$ with $\nu=2$ and the subsystem
\begin{equation}
W_1(\zb) = \frac{b(\zb)}{a(\zb)}= \frac{\sum_{\kb\in\Lambda_+} b_\kb \zb^{-\kb}}{\sum_{\kb\in\Lambda_+} a_\kb \zb^{-\kb}},
\end{equation}
where the index set $\Lambda_+:=\{ (k_1,k_2)\in\Zbb^2 \,:\, 0\leq k_1,k_2\leq1\}$ so that $W_1$ (roughly) has order one in each dimension, the indeterminates $(z_1,z_2)$ are abbreviated as $\zb$, and $\zb^\kb$ stands for $z_1^{k_1} z_2^{k_2}$. 
For notational convenience, the value of the polynomial $a(\zb)$ on the unit torus is written as $a(\thetab)\equiv a(e^{i\theta_1}, e^{i\theta_2})$.
%If the white noise input $e_{\tb}$ has unit variance, then the spectral density of the output process $y_{\tb}$ is $\Phi(\thetab)=[P(\thetab)/Q(\thetab)]^\nu$ where
%$P(\thetab) = |b(\thetab)|^2$, $Q(\thetab) = |a(\thetab)|^2$.
%We shall take the integer $\nu=2$ which is known by our identification procedure.
For simplicity, we also impose a separable form
$(1-r_{a,1} z_1^{-1})(1-r_{a,2} z_2^{-1})$ 
on the polynomials $a$, $b$, and take $|r_{a,j}|<1$ for $j=1,2$. % so that the resulting transfer function $W$ is stable and minimum-phase.
The system parameters can be assigned via
\begin{equation*}
\bmat a_{0,0}&a_{0,1}&a_{1,0}&a_{1,1} \emat = \bmat 1&-r_{a,2}&-r_{a,1}&r_{a,1}r_{a,2} \emat.
\end{equation*}
It is convenient to collect the $2$-d system parameters into matrices. In our particular example, we take real parameters%\bin{[TRUE parameters]}
\begin{equation}\label{true_sys_paremters}
A = \bmat 1&-0.7\\-0.5&0.35\emat, \quad
B = \bmat 0.6696&-0.5357\\-0.4018&0.3214\emat
\end{equation}
where $a_{k_1,k_2}=A(k_1+1,k_2+1)$ and similar for $b_{k_1,k_2}$. Notice that the constraint $p_{\zerob}=1$ translates into a normalization condition
$\|B\|_\F=1$ where the subscript $_\F$ denotes the Frobenius norm. The true spectrum $\Phi(\thetab)$ defined on the $2$-d domain is shown in Fig.~\ref{subfig:true_spec}.

%\begin{figure}
%	\begin{center}
%		\includegraphics[width=2.8cm]{true_spec_data_100}  % The printed column width is 8.4 cm.
%		\includegraphics[width=2.8cm]{recon_spec_data_100}  
%		\includegraphics[width=2.8cm]{abs_err_spec_data_100}  
%		\caption{{\emph{Left.} The true spectrum $\Phi$ of $y$. \emph{Center.} Estimated spectrum $\hat{\Phi}$ via solving the regularized dual optimization problem with $\lambda=10^{-10}$ and \emph{true} covariances and generalized cepstral coefficients. \emph{Right.} The pointwise relative error.}} 
%		\label{fig:true_spec}
%	\end{center}
%\end{figure}

\begin{figure}
%	\centering
	\begin{subfigure}[b]{.48\columnwidth}
		\includegraphics[width=\linewidth]{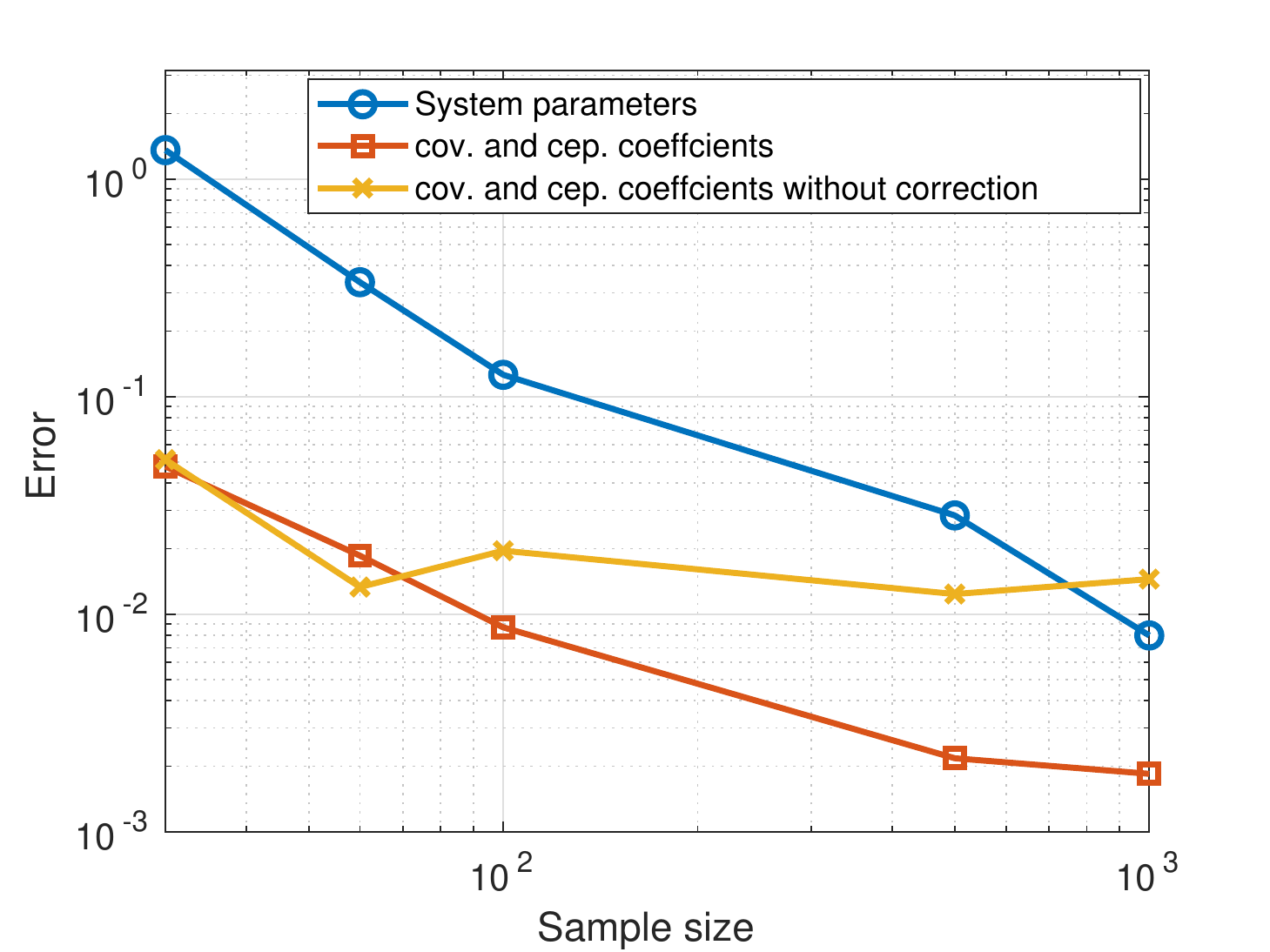}
		\caption{}
		\label{subfig:err_sys_para}
	\end{subfigure}
    \hfill
	\begin{subfigure}[b]{.48\columnwidth}
		\includegraphics[width=\linewidth]{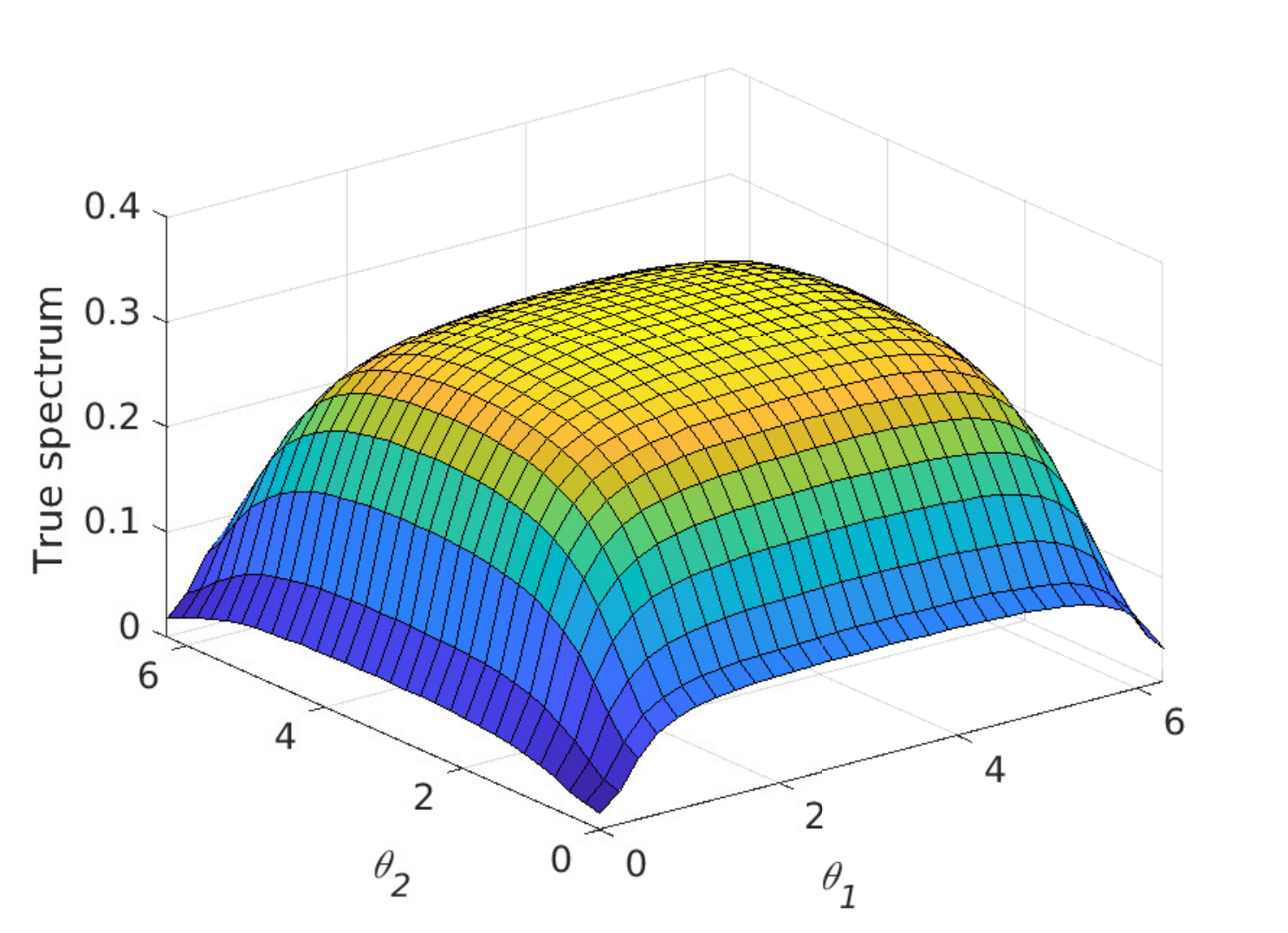}
		\caption{}
		\label{subfig:true_spec}
	\end{subfigure}

	\begin{subfigure}[b]{.48\linewidth}
		\includegraphics[width=\linewidth]{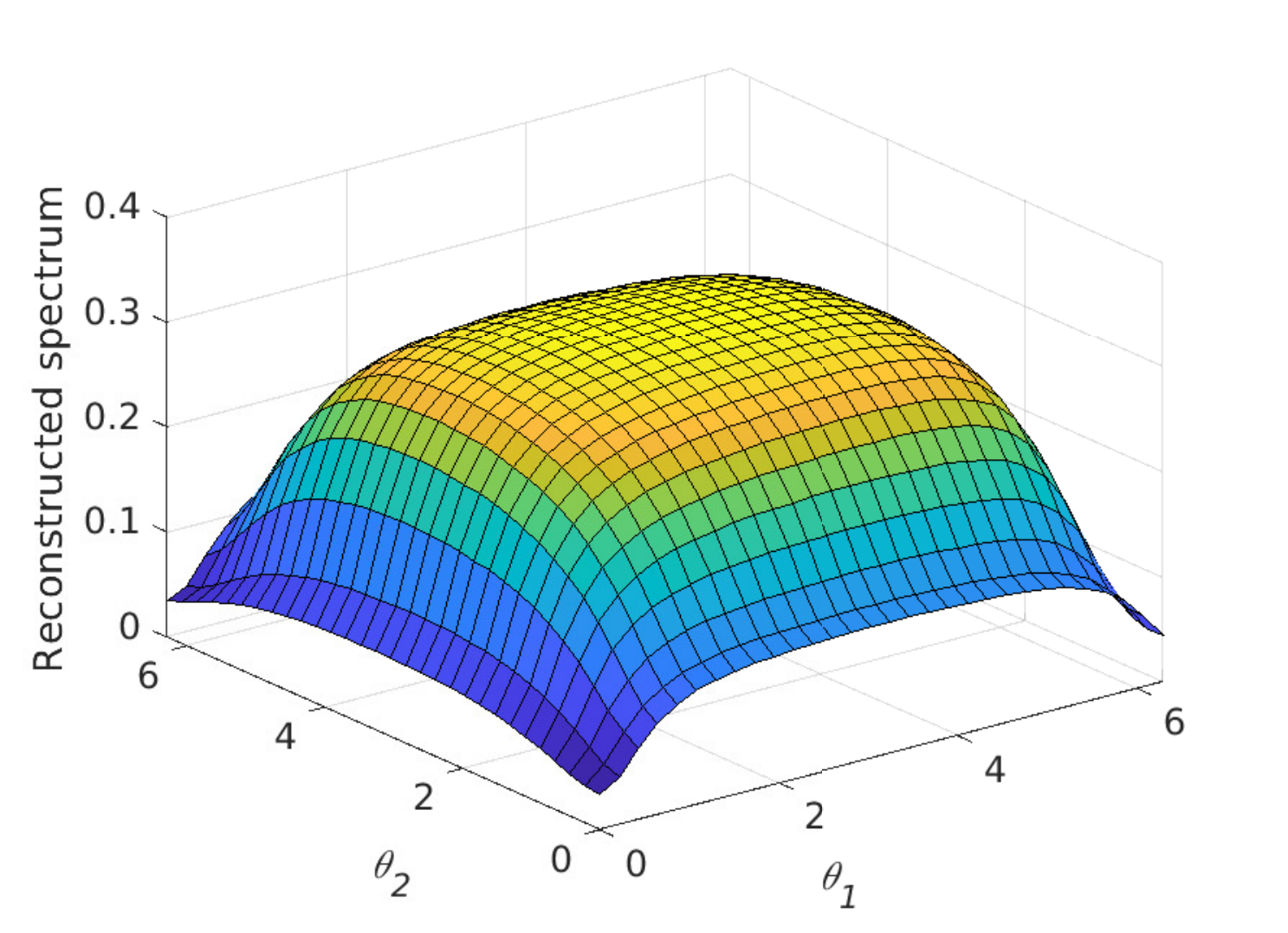}
		\caption{}
		\label{subfig:recon_spec}
	\end{subfigure}
    \hfill
	\begin{subfigure}[b]{.48\linewidth}
		\includegraphics[width=\linewidth]{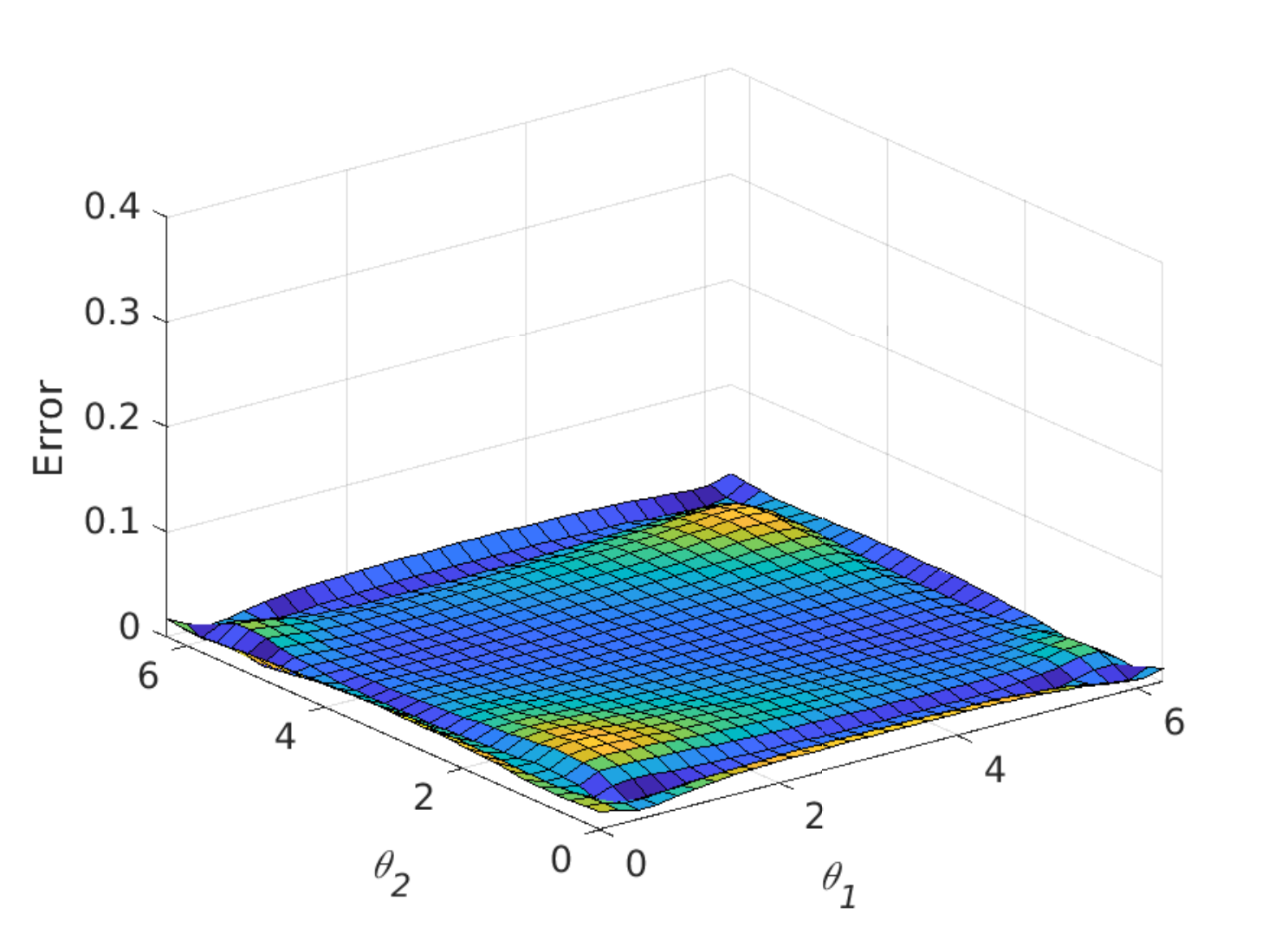}
		\caption{}
		\label{subfig:abs_err}
	\end{subfigure}

	\caption{Simulation results in the $2$-d case: \emph{Upper left.} Estimation errors of the covariances and the generalized cepstral coefficients (the orange line with correction and the yellow line without correction) and the system parameters (blue line) when the sample size $N_1=N_2=N$ of the random field is equal to $30, 60, 100, 500, 1000$. \emph{Upper right.} The true spectrum $\Phi(\thetab)$ of $y_{\tb}$. \emph{Lower left.} The estimated spectrum $\hat{\Phi}(\thetab)$ via solving the discrete version of the regularized dual optimization problem \eqref{reg_dual_pb} with $\lambda=10^{-6}$ and the sample size $N=100$. \emph{Lower right.} The pointwise absolute error $|\hat{\Phi}(\thetab) - \Phi(\thetab)|$ between Subfigs.~\ref{subfig:true_spec} and \ref{subfig:recon_spec}.}
	\label{fig:sim_results}
\end{figure}

%The system is excited by a white noise process/field $e_{t_1,t_2}$ and produces an output process $y_{t_1,t_2}$, see Fig.~\ref{fig:2-d_linear_system}.

Next we implement the system identification procedure in order to estimate the ``system matrices'' $A$ and $B$ from the samples of the output process $y_{\tb}$ with a size $\Nb=(N_1, N_2)$. %Such samples of size $\Nb=(N_1, N_2)$ are generated by feeding the normalized \emph{real} white noise $e_{\tb}$ to the true system (so $y_{\tb}$ is also real).
The covariances and generalized cepstral coefficients of $y_{\tb}$ indexed by the set $\Lambda=\{ (k_1,k_2)\in\Zbb^2 \,:\, -1\leq k_1,k_2\leq1\}$ are estimated using the unwindowed periodogram as explained in Section~\ref{sec:multidim}. Due to the symmetry of $c_{\kb}$ and $m_{\kb}$ with respect to the origin, we only need to compute
\begin{equation*}
\begin{aligned}
\hat\cb  & = \bmat \hat c_{0,0}& \hat c_{0,1}& \hat c_{1,-1}& \hat c_{1,0}& \hat c_{1,1} \emat, \\
\hat\mb  & = \bmat \hat m_{0,1}& \hat m_{1,-1}& \hat m_{1,0}& \hat m_{1,1} \emat
\end{aligned}
\end{equation*}
which are put in the lexicographic ordering. In particular, the constant of correction $C$ in \eqref{ms_converg_hat_mk_multidim} is now $1/\Gamma(\frac{3}{2}) = 2/\sqrt{\pi}\approx1.1284$ since $\alpha=1-\frac{1}{\nu}=1/2$. The estimation error $\|[\hat{\cb}, \hat{\mb}]-[\cb, \mb]\|$ versus the sample size $N_1=N_2=N$ is plotted in the orange line of Fig.~\ref{subfig:err_sys_para} which shares the same trend as that in Fig.~\ref{subfig:err_sys_para_1d}: the error goes down as the sample size increases. The yellow line corresponding to the wrong generalized cepstral estimator, in contrast, deviates from the orange line as early as $N=100$ in this example. Obviously, the wrong estimator does not enjoy the consistency property since the error curve stops going down as $N$ further increases.

Then we solve the regularized dual problem \eqref{reg_dual_pb} on a discrete grid of \emph{fixed} size $\Kb=(30,30)$ with the estimated $\hat{\cb}$, $\hat{\mb}$, and $\lambda=10^{-6}$. % using a gradient descent algorithm.
%Clearly, the variables $\pb$ and $\qb$ can be arranged in the same way as that for $\hat{\mb}$ and $\hat{\cb}$ so we have $9$ real variables in total (a small-scale problem). 
%The algorithm is initialized at $q_{0,0}=1$ and the rest variables equal to $0$ which corresponds to constant polynomials $P(\thetab)=Q(\thetab)\equiv1$. 
%The iterations terminate when the norm of the gradient is less than $10^{-4}$. 
The optimal spectrum $\hat{\Phi}(\thetab)$ returned by the solver is shown in Fig.~\ref{subfig:recon_spec} when the sample size is $N_1=N_2=100$. In this case, clearly the shape of the estimated spectrum is almost visually indistinguishable from that of the true spectrum. Moreover, the pointwise absolute error $|\hat{\Phi}(\thetab)-\Phi(\thetab)|$ is shown in Fig.~\ref{subfig:abs_err}. One can see that the error plot is quite flat. As a complement, we also compute the cumulative relative error on the whole grid $\|\hat{\Phi}-\Phi\|_\F/\|\Phi\|_\F=4.40\%$ which is reasonably small.

Once we have the optimal polynomials $\hat{P}$ and $\hat{Q}$, we can run the factorization algorithm in \cite[Theorems 1.1.1 \& 1.1.3]{geronimo2004positive} and compute factors $\hat{a}$ and $\hat{b}$. 
The numerical values of their coefficients in this particular instance with a sample size $N=100$ are reported in the following matrices %\bin{[ESTIMATED parameters]}
\begin{equation*}
\hat{A} = \bmat 1.0590 & -0.6896 \\ -0.4503 & 0.2854\emat, \quad
\hat{B} = \bmat 0.7102 & -0.5282 \\ -0.3679 & 0.2688\emat.
\end{equation*}
The error here is $\|[\hat{A}, \hat{B}]-[A, B]\|_\F = 0.1259$ which, along with other errors for different $N$, is depicted in the blue line of Fig.~\ref{subfig:err_sys_para}. %\bin{[TBC...]}
%we can conclude that the system parameters are recovered with very small errors, 
We can see the consistency of our identification procedure as reported in Proposition~\ref{prop_identifiability}.
%The blue line shares the same trend with the orange line as the sample size $N$ increases, which is an expected result since the true system $W(\zb)$ belongs to model class dictated by the solution form \eqref{Phi_optimal} of our optimization problem, as explained in Proposition~\ref{prop_identifiability}.

\section{Conclusions}\label{sec:conclusions}

%{\mz 
We have proposed an estimator for the generalized cepstral coefficients both in the unidimensional and multidimensional setting. We showed that it is consistent in the {special} case that the stochastic process is 
i.i.d. or periodic. We also provided a condition that guarantees the consistency in the {general} case {where} the spectral components are correlated. Such a condition seems to be satisfied {when} the stochastic process is generated by a stable rational filter.  Finally, we showed that this estimator can be used to build a consistent estimator for {a class of} cascade linear stochastic systems.
%}

%\appendix
% If any...


\bibliographystyle{IEEEtran}
%\bibliography{biblio-FreqEst}
\bibliography{references}

\begin{thebibliography}{10}
\providecommand{\url}[1]{#1}
\csname url@samestyle\endcsname
\providecommand{\newblock}{\relax}
\providecommand{\bibinfo}[2]{#2}
\providecommand{\BIBentrySTDinterwordspacing}{\spaceskip=0pt\relax}
\providecommand{\BIBentryALTinterwordstretchfactor}{4}
\providecommand{\BIBentryALTinterwordspacing}{\spaceskip=\fontdimen2\font plus
\BIBentryALTinterwordstretchfactor\fontdimen3\font minus
  \fontdimen4\font\relax}
\providecommand{\BIBforeignlanguage}[2]{{%
\expandafter\ifx\csname l@#1\endcsname\relax
\typeout{** WARNING: IEEEtran.bst: No hyphenation pattern has been}%
\typeout{** loaded for the language `#1'. Using the pattern for}%
\typeout{** the default language instead.}%
\else
\language=\csname l@#1\endcsname
\fi
#2}}
\providecommand{\BIBdecl}{\relax}
\BIBdecl

\bibitem{BGL-98}
C.~I. Byrnes, S.~V. Gusev, and A.~Lindquist, ``A convex optimization approach
  to the rational covariance extension problem,'' \emph{SIAM Journal on Control
  Optimization}, vol.~37, no.~1, pp. 211--229, 1998.

\bibitem{SIGEST-01}
------, ``From finite covariance windows to modeling filters: {A} convex
  optimization approach,'' \emph{SIAM Review}, vol.~43, no.~4, pp. 645--675,
  2001.

\bibitem{Georgiou-01}
T.~T. Georgiou, ``Spectral estimation via selective harmonic amplification,''
  \emph{IEEE Transactions on Automatic Control}, vol.~46, no.~1, pp. 29--42,
  2001.

\bibitem{byrnes2001cepstral}
C.~Byrnes, P.~Enqvist, and A.~Lindquist, ``Cepstral coefficients, covariance
  lags, and pole-zero models for finite data strings,'' \emph{IEEE Transactions
  on Signal Processing}, vol.~49, no.~4, pp. 677--693, 2001.

\bibitem{FPR-08}
A.~Ferrante, M.~Pavon, and F.~Ramponi, ``Hellinger versus {K}ullback--{L}eibler
  multivariable spectrum approximation,'' \emph{IEEE Transactions on Automatic
  Control}, vol.~53, no.~4, pp. 954--967, 2008.

\bibitem{FRT-11}
A.~Ferrante, F.~Ramponi, and F.~Ticozzi, ``On the convergence of an efficient
  algorithm for {K}ullback--{L}eibler approximation of spectral densities,''
  \emph{IEEE Transactions on Automatic Control}, vol.~56, no.~3, pp. 506--515,
  2011.

\bibitem{RFP-09}
F.~Ramponi, A.~Ferrante, and M.~Pavon, ``A globally convergent matricial
  algorithm for multivariate spectral estimation,'' \emph{IEEE Transactions on
  Automatic Control}, vol.~54, no.~10, pp. 2376--2388, 2009.

\bibitem{ringh2018multidimensional}
A.~Ringh, J.~Karlsson, and A.~Lindquist, ``Multidimensional rational covariance
  extension with approximate covariance matching,'' \emph{SIAM Journal on
  Control and Optimization}, vol.~56, no.~2, pp. 913--944, 2018.

\bibitem{zhu2018wellposed}
B.~Zhu, ``On the well-posedness of a parametric spectral estimation problem and
  its numerical solution,'' \emph{IEEE Transactions on Automatic Control},
  vol.~65, no.~3, pp. 1089--1099, 2020.

\bibitem{FMP-12}
A.~Ferrante, C.~Masiero, and M.~Pavon, ``Time and spectral domain relative
  entropy: A new approach to multivariate spectral estimation,'' \emph{IEEE
  Transactions on Automatic Control}, vol.~57, no.~10, pp. 2561--2575, 2012.

\bibitem{Z-14}
M.~Zorzi, ``A new family of high-resolution multivariate spectral estimators,''
  \emph{IEEE Transactions on Automatic Control}, vol.~59, no.~4, pp. 892--904,
  2014.

\bibitem{Zhu-Baggio-19}
B.~Zhu and G.~Baggio, ``On the existence of a solution to a spectral estimation
  problem \emph{\`a la} {B}yrnes-{G}eorgiou-{L}indquist,'' \emph{IEEE
  Transactions on Automatic Control}, vol.~64, no.~2, pp. 820--825, 2019.

\bibitem{ephraim1999second}
Y.~Ephraim and M.~Rahim, ``On second-order statistics and linear estimation of
  cepstral coefficients,'' \emph{IEEE Transactions on Speech and Audio
  Processing}, vol.~7, no.~2, pp. 162--176, 1999.

\bibitem{ephraim2005second}
Y.~Ephraim and W.~J. Roberts, ``On second-order statistics of log-periodogram
  with correlated components,'' \emph{IEEE Signal Processing Letters}, vol.~12,
  no.~9, pp. 625--628, 2005.

\bibitem{enqvist2004aconvex}
P.~Enqvist, ``A convex optimization approach to {ARMA}$(n,m)$ model design from
  covariance and cepstral data,'' \emph{SIAM Journal on Control and
  Optimization}, vol.~43, no.~3, pp. 1011--1036, 2004.

\bibitem{Zhu-Zorzi-2021-cepstral}
B.~Zhu and M.~Zorzi, ``Multidimensional rational covariance and cepstral
  extension: A general formulation,'' Submitted to SIAM Journal on Control and
  Optimization, 2021.

\bibitem{RFP-10-wellposedness}
F.~Ramponi, A.~Ferrante, and M.~Pavon, ``On the well-posedness of multivariate
  spectrum approximation and convergence of high-resolution spectral
  estimators,'' \emph{Systems \& Control Letters}, vol.~59, no.~3, pp.
  167--172, 2010.

\bibitem{priestley1981spectral}
M.~B. Priestley, \emph{Spectral Analysis and Time Series (Two-Volume Set)},
  ser. Probability and Mathematical Statistics.\hskip 1em plus 0.5em minus
  0.4em\relax Elsevier Academic Press, 1981, reprinted in 2004.

\bibitem{KLR-16multidimensional}
J.~Karlsson, A.~Lindquist, and A.~Ringh, ``The multidimensional moment problem
  with complexity constraint,'' \emph{Integral Equations and Operator Theory},
  vol.~84, no.~3, pp. 395--418, 2016.

\bibitem{Zhu-Zorzi-21}
B.~Zhu and M.~Zorzi, ``A generalized multidimensional circulant rational
  covariance and cepstral extension problem,'' in \emph{IFAC PapersOnLine},
  vol.~54, no.~7.\hskip 1em plus 0.5em minus 0.4em\relax Padova, Italy: IFAC,
  2021, pp. 553--558, presented virtually at the 19th IFAC Symposium on System
  Identification (SYSID 2021).

\bibitem{ringh2015multidimensional}
A.~Ringh, J.~Karlsson, and A.~Lindquist, ``The multidimensional circulant
  rational covariance extension problem: Solutions and applications in image
  compression,'' in \emph{54th Annual Conference on Decision and Control
  (CDC)}.\hskip 1em plus 0.5em minus 0.4em\relax IEEE, 2015, pp. 5320--5327.

\bibitem{tokuda1990generalized}
K.~Tokuda, T.~Kobayashi, and S.~Imai, ``Generalized cepstral analysis of
  speech-unified approach to {LPC} and cepstral method,'' in \emph{First
  International Conference on Spoken Language Processing}, 1990.

\bibitem{wahlberg2009variance}
B.~Wahlberg, H.~Hjalmarsson, and J.~M{\aa}rtensson, ``Variance results for
  identification of cascade systems,'' \emph{Automatica}, vol.~45, no.~6, pp.
  1443--1448, 2009.

\bibitem{sandberg2014approximative}
H.~Sandberg, P.~H{\"a}gg, and B.~Wahlberg, ``Approximative model reconstruction
  of cascade systems,'' \emph{Systems \& Control Letters}, vol.~69, pp. 90--97,
  2014.

\bibitem{LP15}
A.~Lindquist and G.~Picci, \emph{Linear Stochastic Systems: A Geometric
  Approach to Modeling, Estimation and Identification}, ser. Series in
  Contemporary Mathematics.\hskip 1em plus 0.5em minus 0.4em\relax
  Springer-Verlag Berlin Heidelberg, 2015, vol.~1.

\bibitem{stoica2005spectral}
P.~Stoica and R.~Moses, \emph{Spectral Analysis of Signals}.\hskip 1em plus
  0.5em minus 0.4em\relax Upper Saddle River, NJ: Pearson Prentice Hall, 2005.

\bibitem{Gray-2006}
R.~M. Gray, ``Toeplitz and circulant matrices: A review,'' \emph{Foundations
  and Trends in Communications and Information Theory}, vol.~2, no.~3, pp.
  155--239, 2006.

\bibitem{LPcirculant-13}
A.~Lindquist and G.~Picci, ``The circulant rational covariance extension
  problem: The complete solution,'' \emph{IEEE Transactions on Automatic
  Control}, vol.~58, no.~11, pp. 2848--2861, 2013.

\bibitem{gradshteyn2014table}
I.~S. Gradshteyn and I.~M. Ryzhik, \emph{Table of Integrals, Series, and
  Products}, 8th~ed.\hskip 1em plus 0.5em minus 0.4em\relax Academic Press,
  2014.

\bibitem{Villani-85}
A.~Villani, ``Another note on the inclusion ${L}^p(\mu) \subset {L}^q(\mu)$,''
  \emph{The American Mathematical Monthly}, vol.~92, no.~7, pp. 485--487, 1985.

\bibitem{abramowitz1972handbook}
M.~Abramowitz and I.~A. Stegun, Eds., \emph{Handbook of Mathematical Functions
  with Formulas, Graphs, and Mathematical Tables}, 10th~ed., ser. National
  Bureau of Standards Applied Mathematics Series.\hskip 1em plus 0.5em minus
  0.4em\relax U.S. Government Printing Office, 1972, vol.~55, reprinted by
  Dover Publications in 2020 with corrections.

\bibitem{davies1980algorithm}
R.~B. Davies, ``Algorithm {AS}155: The distribution of a linear combination of
  $\chi^2$ random variables,'' \emph{Applied Statistics}, pp. 323--333, 1980.

\bibitem{engels2017advances}
F.~Engels, P.~Heidenreich, A.~M. Zoubir, F.~K. Jondral, and M.~Wintermantel,
  ``Advances in automotive radar: A framework on computationally efficient
  high-resolution frequency estimation,'' \emph{IEEE Signal Processing
  Magazine}, vol.~34, no.~2, pp. 36--46, 2017.

\bibitem{ZFKZ2019fusion}
B.~Zhu, A.~Ferrante, J.~Karlsson, and M.~Zorzi, ``Fusion of sensors data in
  automotive radar systems: A spectral estimation approach,'' in \emph{58th
  IEEE Conference on Decision and Control (CDC 2019)}.\hskip 1em plus 0.5em
  minus 0.4em\relax IEEE, 2019, pp. 5088--5093.

\bibitem{renken2014microstructured}
A.~Renken, M.~N. Kashid, and L.~Kiwi-Minsker, \emph{Microstructured Devices for
  Chemical Processing}.\hskip 1em plus 0.5em minus 0.4em\relax John Wiley \&
  Sons, 2014.

\bibitem{vavilin2007anaerobic}
V.~Vavilin, L.~Lokshina, X.~Flotats, and I.~Angelidaki, ``Anaerobic digestion
  of solid material: Multidimensional modeling of continuous-flow reactor with
  non-uniform influent concentration distributions,'' \emph{Biotechnology and
  Bioengineering}, vol.~97, no.~2, pp. 354--366, 2007.

\bibitem{burg1975maximum}
J.~P. Burg, ``Maximum entropy spectral analysis,'' Ph.D. dissertation,
  Department of Geophysics, Stanford University, 1975.

\bibitem{Z-14rat}
M.~Zorzi, ``Rational approximations of spectral densities based on the {A}lpha
  divergence,'' \emph{Mathematics of Control, Signals, and Systems}, vol.~26,
  no.~2, pp. 259--278, 2014.

\bibitem{zorzi2015interpretation}
------, ``An interpretation of the dual problem of the {THREE}-like
  approaches,'' \emph{Automatica}, vol.~62, pp. 87--92, 2015.

\bibitem{RKL-16multidimensional}
A.~Ringh, J.~Karlsson, and A.~Lindquist, ``Multidimensional rational covariance
  extension with applications to spectral estimation and image compression,''
  \emph{SIAM Journal on Control and Optimization}, vol.~54, no.~4, pp.
  1950--1982, 2016.

\bibitem{ZFKZ2019M2}
B.~Zhu, A.~Ferrante, J.~Karlsson, and M.~Zorzi, ``M$^2$-spectral estimation: A
  relative entropy approach,'' \emph{Automatica}, vol. 125, March 2021.

\bibitem{dritschel2004factorization}
M.~A. Dritschel, ``On factorization of trigonometric polynomials,''
  \emph{Integral Equations and Operator Theory}, vol.~49, no.~1, pp. 11--42,
  2004.

\bibitem{geronimo2004positive}
J.~S. Geronimo and H.~J. Woerdeman, ``Positive extensions, {F}ej{\'e}r-{R}iesz
  factorization and autoregressive filters in two variables,'' \emph{Annals of
  Mathematics}, vol. 160, no.~3, pp. 839--906, 2004.

\bibitem{sayed2001survey}
A.~H. Sayed and T.~Kailath, ``A survey of spectral factorization methods,''
  \emph{Numerical Linear Algebra with Applications}, vol.~8, no. 6-7, pp.
  467--496, 2001.

\end{thebibliography}

\end{document}